\documentclass[%
 reprint,
superscriptaddress,
 amsmath,amssymb,
 aps,
pre,
]{revtex4-1}

\RequirePackage{snapshot}
\usepackage[english]{babel}
\usepackage{amsfonts}
\usepackage{amsmath}
\usepackage{amssymb}
\usepackage{amsthm}
\usepackage{graphicx,psfrag}
\usepackage{color}
\usepackage{epigraph}
\usepackage{pxfonts}
\usepackage{stmaryrd}
\usepackage{listings}
\usepackage{diagbox}
\usepackage{makecell}
\usepackage{textcomp}

\usepackage[pdftitle={Percolation Thresholds in Hyperbolic Lattices},
pdfauthor={Stephan Mertens and Cristopher Moore},
pdfcreator={},
pdfsubject={Percolation},
pdfproducer={},
pdfkeywords={Monte Carlo, percolation, thresholds, hyperbolic},hyperfootnotes=false]{hyperref}
\usepackage{url}
\usepackage[noline]{algorithm2e}
\SetArgSty{mbox}
\SetCommentSty{textrm}

\graphicspath{{../figures/}}

\IfFileExists{date.tex}{%
  \input{date.tex}

}

\newcommand{\D}{\mathrm{d}}

\newcommand*{\Z}{\mathbb{Z}}
\newcommand*{\id}{\mathbb{I}}

\newcommand{\eps}{\varepsilon}

\newcommand{\pcbond}{p_c^{\text{bond}}}
\newcommand{\pubond}{p_u^{\text{bond}}}
\newcommand{\pcsite}{p_c^{\text{site}}}
\newcommand{\pusite}{p_u^{\text{site}}}
\newcommand{\pmatch}{p_c^{\text{match}}}

\DeclareMathOperator{\pred}{predecessor}
\DeclareMathOperator{\suc}{successor}
\DeclareMathOperator{\parent}{parent}
\DeclareMathOperator{\child}{child}
\DeclareMathOperator{\outdegree}{outdegree}
\DeclareMathOperator{\vertextype}{type}

\theoremstyle{plain}
\newtheorem{theorem}{Theorem}

\begin{document}

\title{Percolation Thresholds in Hyperbolic Lattices}

\author{Stephan Mertens}
\email{mertens@ovgu.de}
\affiliation{Santa Fe Institute, 1399 Hyde Park Rd., Santa Fe, NM 87501, USA}
\affiliation{Institut f\"ur Theoretische Physik, Universit\"at
  Magdeburg, Universit\"atsplatz~2, 39016~Magdeburg, Germany}

\author{Cristopher Moore}
\email{moore@santafe.edu}
\affiliation{Santa Fe Institute, 1399 Hyde Park Rd., Santa Fe, NM 87501, USA}

\date{\today}

\begin{abstract}
  We use invasion percolation to compute numerical values for bond and
  site percolation thresholds $p_c$ (existence of an infinite cluster) and $p_u$ (uniqueness of the infinite cluster) of
  tesselations $\{P,Q\}$ of the hyperbolic plane, where $Q$ faces
  meet at each vertex and each face is a $P$-gon. Our values are accurate to six or seven
  decimal places, allowing us to explore their functional dependency on $P$ and $Q$ and
  to numerically compute critical exponents.  We also prove rigorous
  upper and lower bounds for $p_c$ and $p_u$ that can be used to find
  the scaling of both thresholds as a function of $P$ and $Q$.
\end{abstract}

\pacs{64.60.ah,02.70.-c, 02.70.Rr, 05.10.Ln}

\maketitle

\section{Introduction}
\label{sec:intro}

There is plenty of room in the hyperbolic plane: enough 
to host more parallel lines than those claimed by Euclid's fifth postulate,
and enough to allow an infinite number of tesselations by regular
polygons. Fig.~\ref{fig:examples} shows examples of such tesselations,
drawn in the Poincar\'e disk representation of the hyperbolic plane \footnote{The figures were
  generated by a Java applet from D. Hatch,
  \url{www.plunk.org/~hatch}}.

We consider tilings of a surface where $Q$ regular $P$-gons meet at each vertex, 
and we denote such a tiling by the Schl\"afli symbol
$\{P,Q\}$. The surface is flat
if and only if $(P-2)(Q-2)=4$, which has only three solutions:
$\{3,6\}$ (the triangular lattice), $\{4,4\}$ (the square lattice) and
$\{6,3\}$ (the honeycomb lattice).  When $(P-2)(Q-2) > 4$, 
the surface has negative Gaussian curvature, and in that case we refer to
$\{P,Q\}$ as a hyperbolic lattice.

\begin{figure}[hb]
  \centering
  \includegraphics[width=0.3\columnwidth]{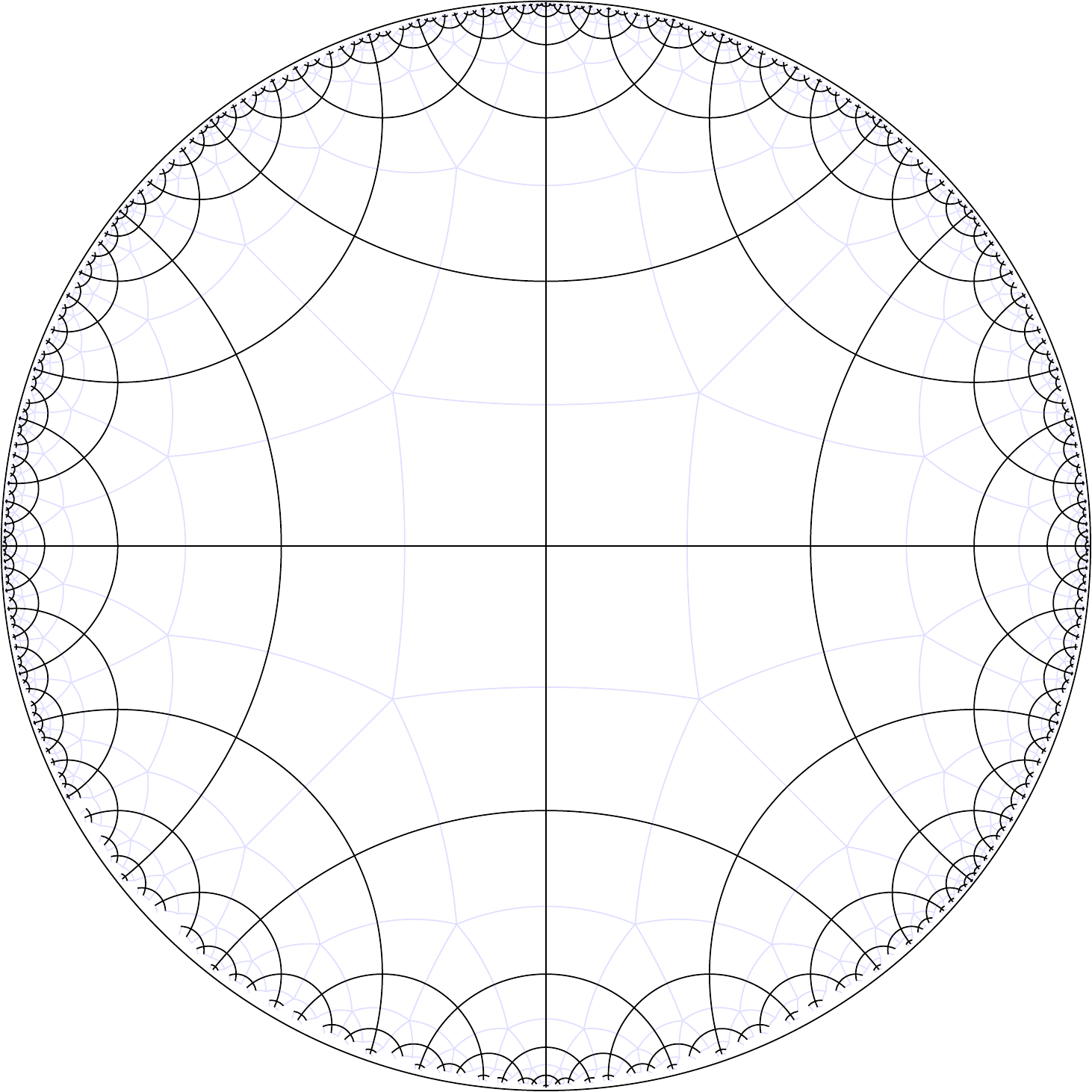}\hfill
  \includegraphics[width=0.3\columnwidth]{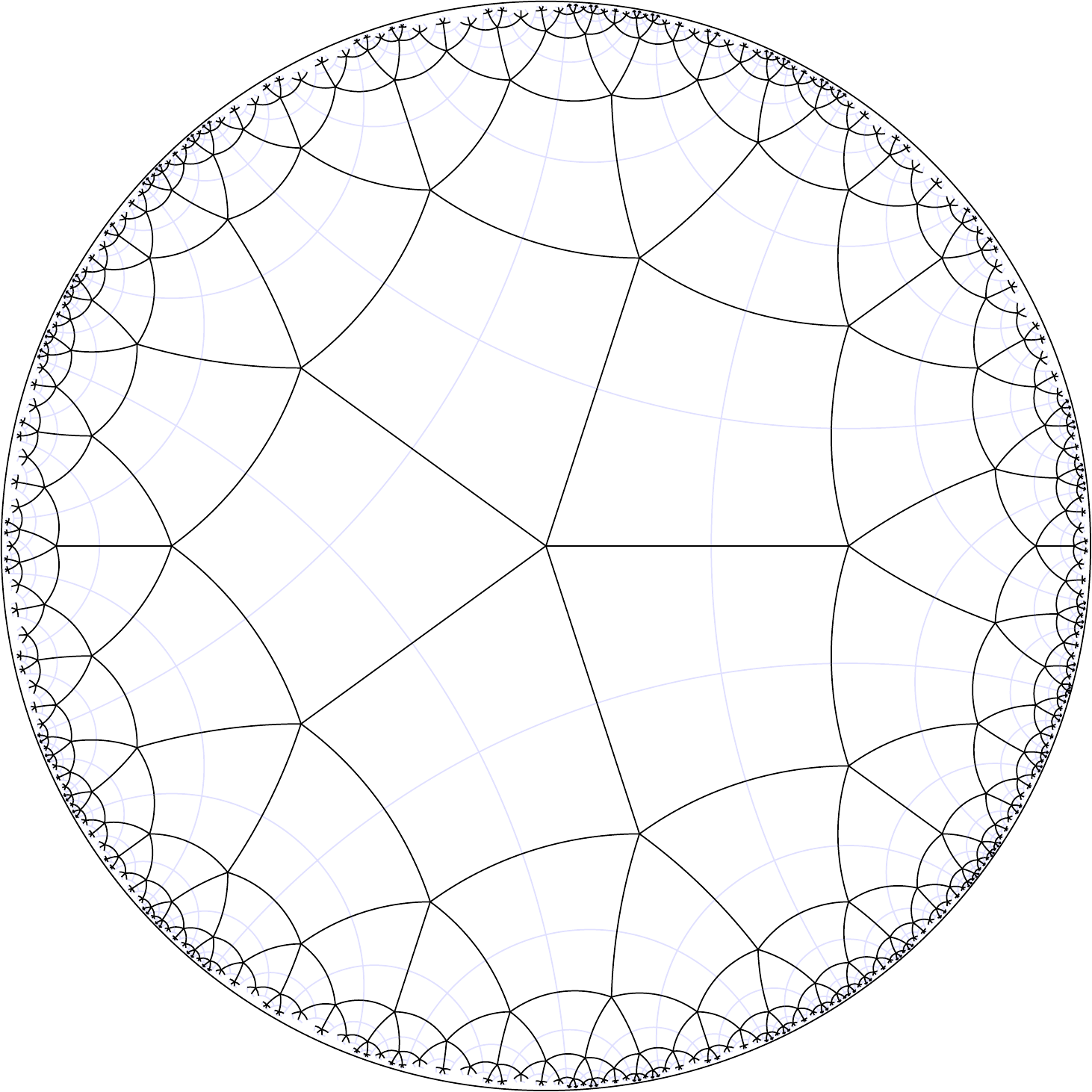}\hfill
  \includegraphics[width=0.3\columnwidth]{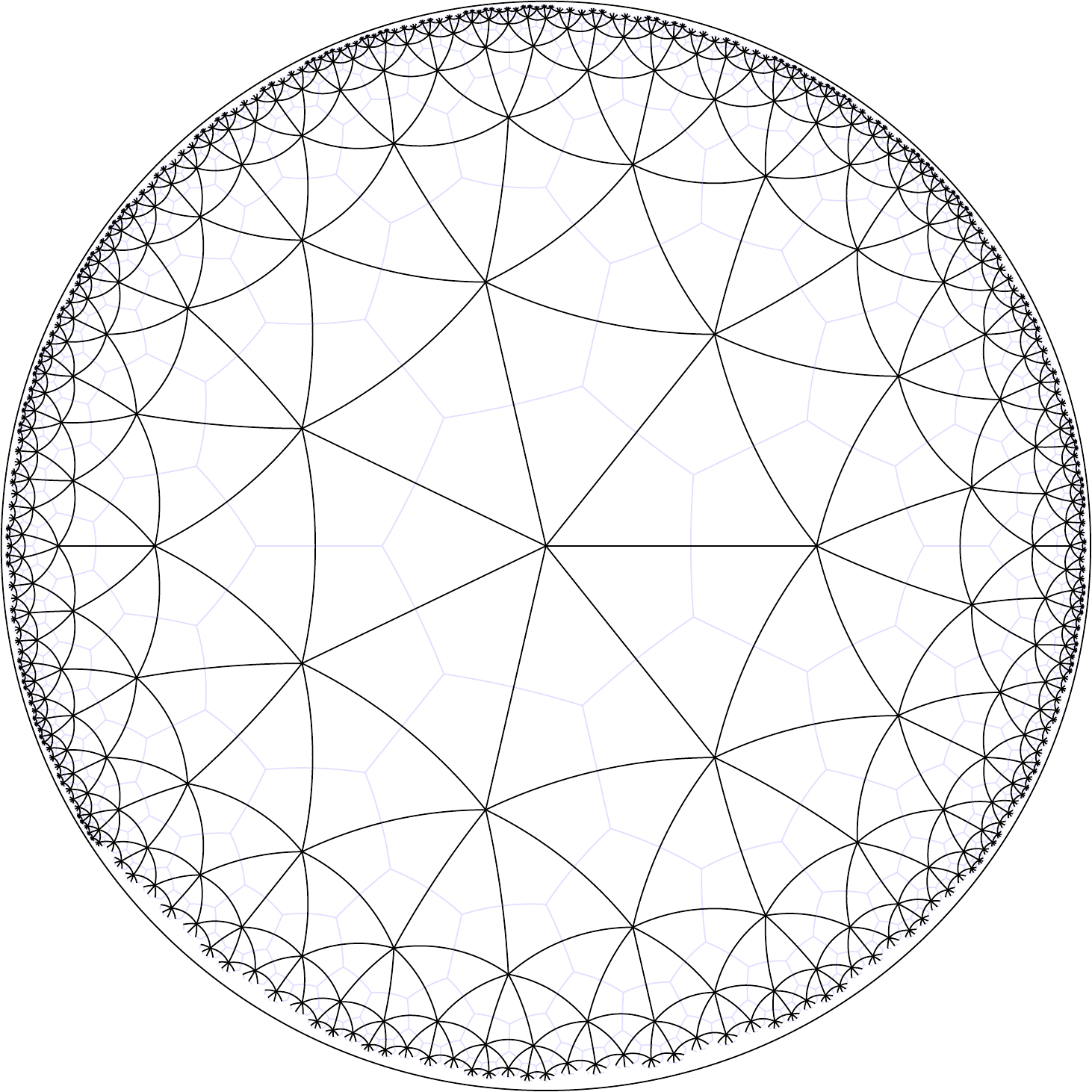}
  \caption{Hyperbolic lattices $\{5,4\}$, $\{4,5\}$ and $\{3,7\}$
    (left to right) }
  \label{fig:examples}
\end{figure}

Hyperbolic lattices, and more generally graphs embedded in hyperbolic space, were first popularized through the art of
M.C. Escher~\cite{schattschneider:04}.  They have also become
popular in physics and computer science, in studies of Brownian 
motion~\cite{monthus:texier:96}, diffusion~\cite{baek:yi:kim:08}, complex
networks~\cite{krioukov:etal:10}, cellular automata~\cite{margenstern}, 
hard disks~\cite{modes:kamien:07}, and the critical
exponents of the Ising model~\cite{shima:sakaniwa:06,shima:sakaniwa:06a}.

In this contribution we will discuss percolation on hyperbolic lattices, and in particular how to compute rigorous bounds on percolation thresholds and highly-accurate numerical values for them.  The latter is a challenge for several reasons.

The first problem is labeling the vertices of a hyperbolic lattice. Unlike in Euclidean lattices such as $\Z^d$, we cannot simply refer to vertices with vectors of integers.  In Appendix~\ref{sec:implementation}, we present a labeling scheme which largely overcomes this difficulty, giving each vertex a unique string over a finite alphabet which makes it easy to generate a list of its neighbors.

The second, and more severe, problem is the exponential growth of the number of vertices as a function of distance from a given vertex (see Appendix~\ref{sec:counting}). This limits the size of lattices that can be stored in a computer, and is probably the main reason why previous numerical measurements of percolation thresholds in hyperbolic lattices are not very accurate: bond percolation thresholds have been calculated to only two decimal places~\cite{baek:minhagen:kim:09,gu:ziff:12}, and site percolation thresholds are simply missing from the literature.  

We avoid the need to store large hyperbolic lattices by using the invasion percolation algorithm~\cite{wilkinson:willemsen:83}, which we review in Section~\ref{sec:invasion}.  Combining this with our labeling scheme allows us to store just the vertices in the percolating cluster, as opposed to the entire lattice.  As a result, we can compute site and bond percolation thresholds to at least six decimal places (see Section~\ref{sec:results}).  This is accurate enough to analyze how these thresholds scale with $P$ and $Q$ and numerically compute critical exponents (Section~\ref{sec:exponents}).  We also prove several rigorous upper and lower bounds on these thresholds, and compare them to our numerical results.

Finally, unlike Euclidean lattices, hyperbolic lattices have two distinct percolation thresholds.  At the first threshold $p_c$, infinite clusters appear, but there are many of them.  At the second threshold $p_u$, they merge to form a single cluster, so that the infinite cluster is unique. \emph{Prima facie} it is not obvious how to compute the uniqueness threshold $p_u$, but we will show that this problem can be mapped to the more familar task of computing $p_c$ on the so-called matching lattice\cite{sykes:essam:64}. For this, we need a bit of theory, which we present in the next section.

\section{The uniqueness threshold and the matching lattice}
\label{sec:theory}

A salient feature of hyperbolic lattices is that they are \emph{nonamenable}. An infinite graph is amenable if the surface-to-volume ratio tends to zero: equivalently, if the volume of a sphere of radius $\ell$, i.e., the number of vertices within $\ell$ steps of a given vertex, grows polynomially rather than exponentially.  Flat lattices with $(P-2)(Q-2)=4$ are
amenable because the volume of a sphere grows as $\ell^2$, and the surface area grows as $\ell$.  In hyperbolic lattices, on the other hand, the volume and surface area of a sphere both grow as $\lambda^\ell$ for the same $\lambda > 1$. Indeed, in the limit $P \to \infty$, the hyperbolic lattice becomes a Bethe lattice or Cayley tree, i.e., an infinite tree where each vertex has $Q-1$ daughters (see e.g.~\cite{ostilli:12,mosseri:sadoc:82}). 

It is known that percolation (site or bond) on planar, nonamenable graphs, including hyperbolic lattices, has two distinct critical densities~\cite{benjamini:schramm:01},
\begin{equation}
  \label{eq:pc-pu}
  0 < p_c < p_u < 1 \, ,
\end{equation}
where $p_c$ is defined as the infimum of $p \in (0,1)$ such that the sites (bonds) selected with probability $p$ form at least one cluster of infinite size, and $p_u$ is defined as the infimum of $p \in (0,1)$ such that the selected sites (bonds) form a \emph{unique} infinite cluster.  Thus for $p < p_c$ there is no infinite cluster, for $p_c < p < p_u$ there are infinitely many infinite clusters, and for $p > p_u$ they have all merged into a single infinite cluster.  In contrast, amenable graphs like $\Z^d$ do not have ''enough room'' to host more than one infinite cluster, so $p_c=p_u$ on these lattices.

There are numerous established numerical methods to measure $p_c$, including union-find algorithms~\cite{newman:ziff:00,mertens:moore:contperc} and exact solutions of small systems~\cite{scullard:jacobsen:12,jacobsen:14,jacobsen:15,mertens:ziff:16}. Detecting the uniqueness of the infinite cluster and thus measuring $p_u$ is \emph{a priori} a different and more daunting task. Fortunately, this problem can be reduced to the problem of computing $p_c$ on a related graph. For planar, nonamenable graphs $G$, it can be proven that 
\begin{equation}
  \label{eq:bond-duality}
  \pubond(G) = 1-\pcbond(G^\dagger) \, ,
\end{equation}
where $G^\dagger$ is the dual of $G$~\cite[Theorem
3.8]{benjamini:schramm:01}. 
The dual of a planar lattice $\{P,Q\}$ is $\{Q,P\}$, so
\begin{equation}
  \label{eq:bond-duality-pq}
  \pubond(\{P,Q\}) = 1 - \pcbond(\{Q,P\}) \, .
\end{equation}
Hence for bond percolation, we only need to solve the familiar problem
of computing $\pcbond$ (on the dual lattice) to get a value for $\pubond$.
But what about site percolation? 

We claim that for planar, nonamenable graphs $G$,
\begin{equation}
  \label{eq:site-duality}
  \pusite(G) = 1 - \pcsite(\hat{G}) \, ,
\end{equation}
where $\hat{G}$ is the \emph{matching lattice} of $G$.  The vertices of $\hat{G}$ are the 
same as those of $G$, but with additional edges so that the vertices around each face
form a clique, a fully connected graph~\cite{sykes:essam:64}. Fig.~\ref{fig:matching-example} shows the matching lattice 
of the $\{4,5\}$ lattice. Note that $G$ is non-amenable if and only
if $\hat{G}$ is non-amenable.

\begin{figure}
  \centering
  \includegraphics[width=\columnwidth]{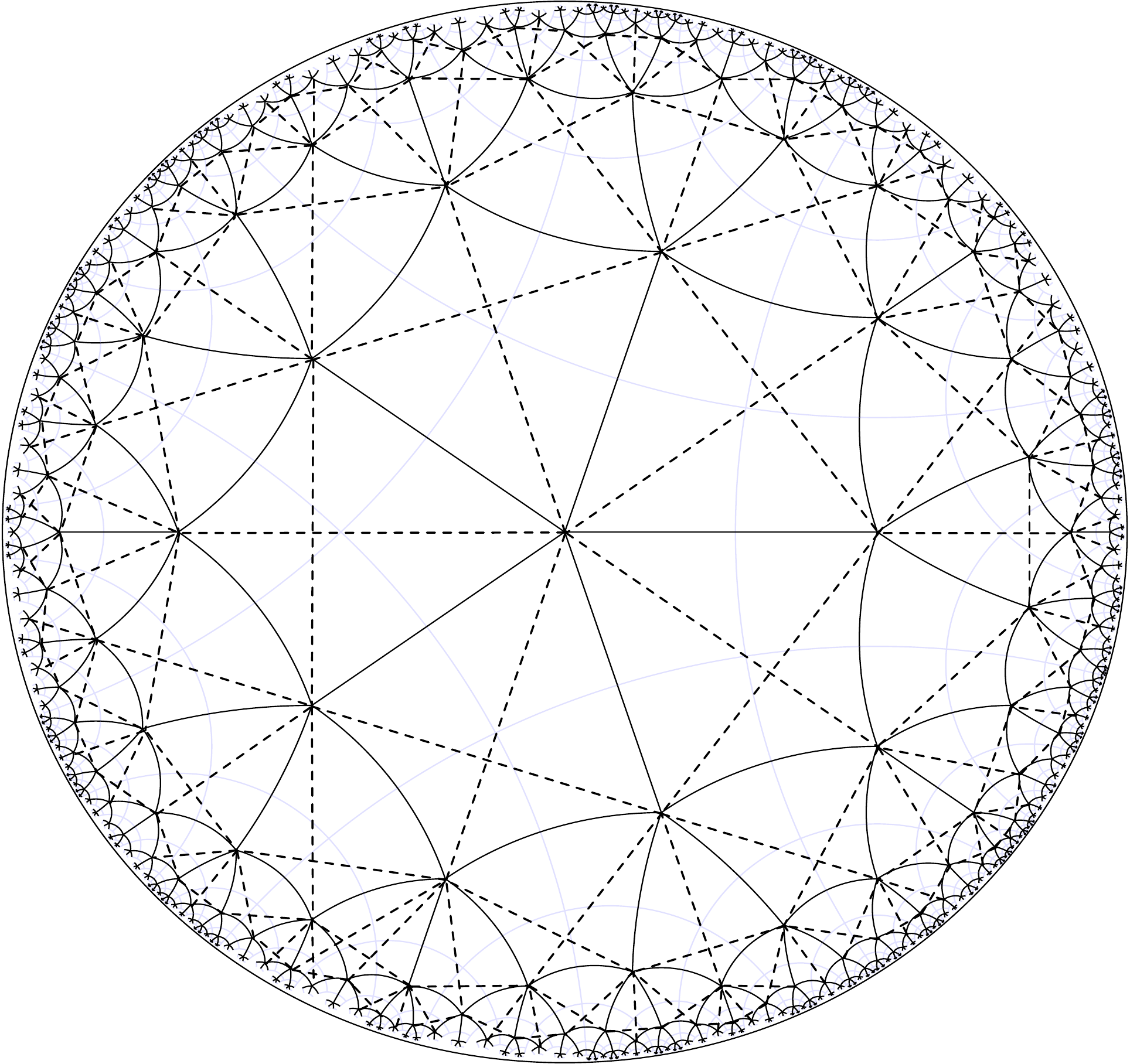}
  \caption{Matching lattice $\hat{G}$ of the $\{4,5\}$ hyperbolic lattice. The new edges, drawn as dashed lines, connect the vertices of each face in a clique.}
  \label{fig:matching-example}
\end{figure}

Now imagine that we color each site black with probability $p$
and white with probability $1-p$.  We connect the black sites through
the edges of $G$, and connect the white sites through the edges of $\hat{G}$.  
Each black component is surrounded by white sites and each white component is
surrounded by black sites. 

The crucial observation is that the white sites surrounding a given
black component are connected in $\hat{G}$;  conversely, the black
sites surrounding a given white component are connected in $G$
(see~\cite[Fig.~1]{mertens:ziff:16} for an example).  Therefore, two
or more infinite black clusters exist if and only if they are
separated by an infinite white cluster, and a unique infinite black cluster exists if and only if there
is no infinite white cluster.  Thus $p < \pcsite(\hat{G})$ implies $1-p > \pusite(G)$ and vice versa, which proves
\eqref{eq:site-duality}.

Note that for the triangular lattices $\{3,Q\}$ we have $G=\hat{G}$, and so
\begin{equation}
  \label{eq:self-matching}
  \pusite(\{3,Q\}) = 1 - \pcsite(\{3,Q\}) \, . 
\end{equation}
For amenable graphs we also have $\pusite = \pcsite$, so that~\eqref{eq:self-matching} gives $\pcsite=1/2$.  This is the well-known result for all planar, amenable graphs with triangular faces.

Having reduced the problem of computing $p_u$ on one lattice
to the problem of computing $p_c$ on another, we
are still faced with the problem of how to compute $p_c$ on exponentially
growing lattices. This is where invasion percolation comes in.



\section{Invasion Percolation and Measuring the Threshold}
\label{sec:invasion}

Invasion percolation is a stochastic growth model that was introduced
as a model for fluid transport through porous media~\cite{lenormand:bories:80,chandler:etal:82,wilkinson:willemsen:83}.  
The growth process starts with a single vertex of the underlying graph
as the seed of the invasion cluster.  In the variant relevant to site percolation, 
we assign each of its neighboring vertices a random weight uniformly distributed between zero and one. 
We then add the neighbor with the smallest weight to the cluster, yielding a 
cluster of mass $N=2$.  This process is iterated: at each step, we assign random weights to each previously 
unassigned vertex in the cluster's neighborhood, and add the
neighboring vertex with the smallest weight to the cluster, increasing
the mass $N$ by $1$. For bond percolation, we assign weights to edges
rather than vertices, and we extend the invasion cluster along the
edge incident to it with the smallest weight.

For our purposes, the benefit of invasion percolation is that we do
not need to store a lattice large enough to hold the largest cluster
that we want to grow: for non-amenable lattices this would be
computationally infeasible. Instead we only need to store the vertices
belonging to the cluster, and the neighboring vertices to which we
have assigned weights. Since the coordination number $Q$ of the
lattice is fixed, the total number of vertices we need to keep track
of grows only linearly with the mass of the cluster.  Since we do not
store a lattice we need to keep track of the vertices that we have
explored so far, and we also need to find the vertex in the
boundary of the cluster with the smallest weight. We use two data
structures to achieve this. A \emph{set} is used to hold all vertices
that have been assigned weights, and a \emph{priority queue} that
holds the hull, i.e. all vertices that have been assigned weights and
that are currently not part of the cluster
\cite{dasgupta:algorithms}. The priority queue lets us select and
remove the lowest-weight vertex from the hull or add new vertices to
it in time logarithmic in the size of the hull. The set lets us remove
and insert vertices or search for a vertex in time logarithmic in the
number of vertices. Modern programming languages provide ready-to-use
implementations of data structures like these.  We used the container
classes \texttt{set} and \texttt{priority\_queue} from the C++
standard library~\cite{vanweert:gregoire:16}.

In order to achieve the logarithmic time complexity of the container
classes, the vertices have to be sortable. Our vertex labeling scheme
for the vertices (Appendix~\ref{sec:implementation}) provides a
natural, lexicographic ordering. Since this scheme requires $O(k)$
memory for a vertex at distance $k$ from the origin, we save memory
(and time) by storing the actual vertices only in the set, whereas the
priority queue holds the weight and a pointer (an iterator in C++
lingo) to the corresponding vertex in the set.

The efficiency of this approach is independent of dimensionality or the exponential growth rate of non-amenable lattices.  One the other hand, it requires a labeling of vertices that makes it easy to compute the labels of its neighbors. For $\Z^d$ this is trivial, but for the hyperbolic lattices $\{P,Q\}$ this is not straightforward. We present
our labeling scheme in Appendix~\ref{sec:implementation}.

\emph{A priori}, invasion percolation differs from classical Bernoulli percolation, where each vertex is independently occupied with probability $p$.  But invasion percolation reproduces, both qualitatively and quantitatively, Bernoulli percolation at criticality~\cite{chayes:chayes:newman:85,haeggstroem:peres:schonmann:99}.  We can think of this connection intuitively as follows.  Since the vertex weights are uniform in the unit interval, one way to implement Bernoulli percolation is to declare a vertex occupied if its weight is less than or equal to $p$.  At $p=p_c$, the occupied sites possess one or more infinite components.  If the initial seed vertex is not in an infinite component, invasion percolation will force its way outward using weights greater than $p_c$; but once the invasion cluster touches an infinite component, it will grow into it, extending the cluster to infinite mass by adding vertices of weight at most $p_c$.  

\begin{figure}
  \centering
  \includegraphics[width=\columnwidth]{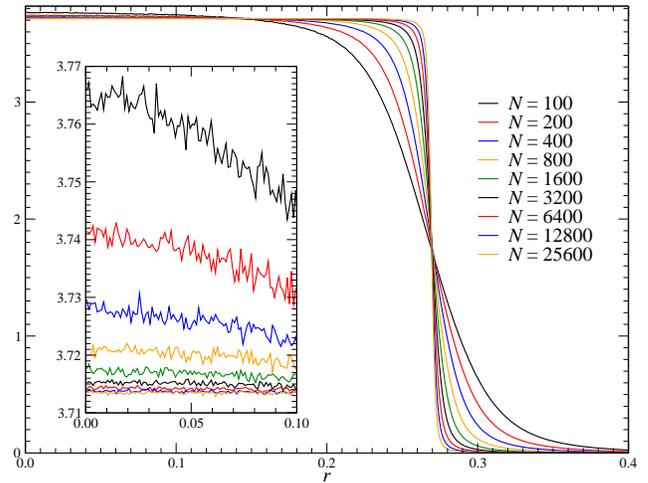}
  \vspace{-24pt}
  \caption{Distribution of weights on site percolation invasion clusters of mass
    $N$ in the $\{3,7\}$ hyperbolic lattice.}
  \label{fig:weights}
\end{figure}

The connection between invasion and percolation suggests that if we keep track of the weights of the vertices added at each step to the invasion cluster, these weights are asymptotically uniform in the interval $[0,p_c]$.  That is, if $w_N(r) \,\D r$ is the fraction of vertices with weight in $(r, r+\D r)$ in an invasion cluster of mass $N$, we have
\begin{equation}
  \label{eq:weight-distribution}
  \lim_{N \to \infty} w_N(r) = \begin{cases}
    1/p_c & 0 \le r \le p_c \\
    0        & p_c < r \le 1 \, .
  \end{cases}
\end{equation}
Fig.~\ref{fig:weights} shows the convergence of the weight distribution $w_N$ to the step function~\eqref{eq:weight-distribution}.  

This suggests several possible estimators of $p_c$.  One is the value of $p$ at the crossover point where $w_N(r)$ drops from $1/p_c$ to $0$, which in Fig.~\ref{fig:weights} takes place at $w_N \approx 1.7$.  Another is to take the reciprocal of $w_N(r)$ for any $r < p_c$.  In particular, the inset in Fig.~\ref{fig:weights} shows $w_N(r)$ for $0 \le r \le 0.1$, which converges to $1/p_c \approx 3.713$ as $N$ increases.

However, in our computations we use another estimator provided by
invasion percolation. Let $B(N)$ denote the number of vertices that
have assigned weights in the course of building a cluster with mass
$N$, i.e., which are either in the cluster or one of its neighbors.
Since almost all of the $N$ vertices actually added to the cluster
have weight less than or equal to $p_c$, we have
\begin{equation}
  \label{eq:volume-surface-ratio}
  \lim_{N \to \infty} \frac{N}{B(N)} = p_c \, .
\end{equation}
This has been proven for bond percolation in $\Z^2$~\cite{chayes:chayes:newman:85}, but it is believed to hold more generally.  Based on this supposition, we use $N/B(N)$ as an estimate of $p_c$. See also~\cite{leath:76a} for a convincing argument as to why $N/B(N)$ is a good estimator for $p_c$ on every lattice.

The estimator $N/B(N)$ is extremely easy to compute, since $N$ and
$B(N)$ are simply integers given by the progress of the invasion
percolation process.  Moreover, it seems to have excellent finite-size
scaling and a very small statistical variaance.  Recall that on a $Q$-regular tree, we have $p_c = 1/(Q-1)$ for both site and bond percolation.  On such a tree, a cluster of mass $N$ is surrounded by exactly $(Q-2)N+2$ neighboring vertices, so
\begin{equation}
  \label{eq:volume-surface-ratio-tree}
  \frac{N}{B(N)} = \frac{N}{(Q-1)N+2} = \frac{p_c}{1+\frac{2}{(Q-1)N}} \, .
\end{equation}
Hence on a tree, the value of $N/B(N)$ does not at all depend on the
random weights chosen to grow the invasion cluster. It is a
deterministic quantity. To some extent, this property is preserved on
hyperbolic lattices: here the standard deviation in $N/B(N)$ is small
and it decays exponentially with $P$, see
Appendix~\ref{sec:technicalities}.
 
To check whether the finite size scaling
\eqref{eq:volume-surface-ratio-tree} also applies to hyperbolic
lattices, we plotted the quantity 
\begin{displaymath}
  \frac{B(N)}{N} - \frac{B(2N)}{2N}
\end{displaymath}
vs. $N$. For the tree~\eqref{eq:volume-surface-ratio-tree}, this quantity is exactly $N^{-1}$. For hyperbolic lattices we find that this quantity scales as $N^{-\delta}$ as shown in Fig.~\ref{fig:check-scaling}.  For finite $P$ we have $\delta < 1$, but Fig.~\ref{fig:delta-convergence} shows that $\delta$ converges to $1$ as $P$ increases, and that this convergence gets faster as $Q$ increases.  Thus $N/B(N)$ converges to $p_c$ almost as quickly as it does on a tree.

\begin{figure}
  \centering
  \includegraphics[width=\columnwidth]{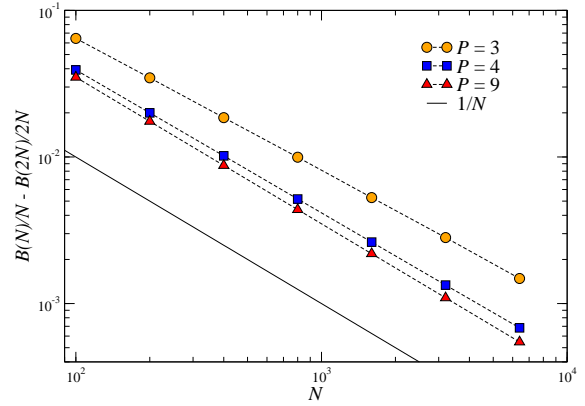}
  \vspace{-24pt}
\caption{$\frac{B(N)}{N}-\frac{B(2N)}{2N} \sim N^{-\delta}$ for bond percolation and $Q=7$.}
  \label{fig:check-scaling}
\end{figure}

\begin{figure}
  \centering
  \includegraphics[width=\columnwidth]{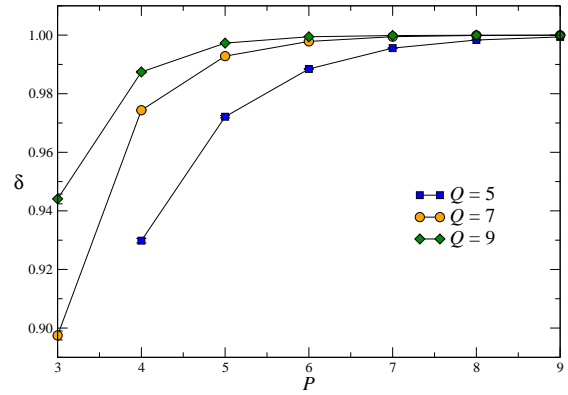}
  \vspace{-24pt}
\caption{Finite size scaling exponent $\delta$ for bond percolation.  Note the convergence to its value $\delta = 1$ for the tree.}
  \label{fig:delta-convergence}
\end{figure}

These results motivated us to fit our numerical data to the form
\begin{equation}
  \label{eq:sv-finite}
  \frac{N}{B(N)} = \frac{p_c}{1 + b N^{-\delta}} \, .
\end{equation}
Using finite-size scaling of this form, we can extrapolate from invasion clusters of size $N = 2^k \cdot 100$ for $k = 0,\ldots,10$ to $N = \infty$.  Fig.~\ref{fig:pc_finite} shows our results for the $\{7,3\}$ lattice.  

\begin{figure}
  \centering
  \includegraphics[width=\columnwidth]{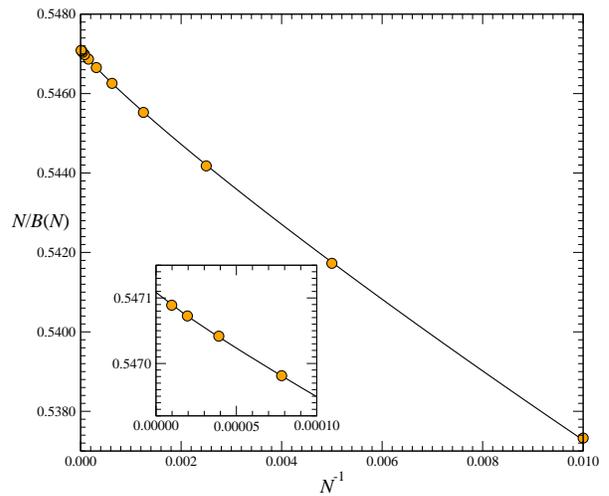}
  \vspace{-24pt}
\caption{The ratio $N/B(N)$ for invasion clusters of mass $N$ for site percolation on the $\{7,3\}$ hyperbolic lattice. The line is a fit of~\eqref{eq:sv-finite}, and extrapolating to $N = \infty$ gives our estimate of $\pcsite$.}
  \label{fig:pc_finite}
\end{figure}

All of the above holds both for bond and site percolation on hyperbolic
lattices as well as on their matching lattices.

\section{Numerical Results and Comparison with Rigorous Bounds}
\label{sec:results}

We ran the invasion percolation algorithm on hyperbolic lattices with
all values of $P, Q \le 9$ and their matching lattices to find
$\pcsite$, $\pusite$ and $\pcbond$. Our results are shown in
Table~\ref{tab:thresholds}, and are accurate to at least six
digits. We have listed all previously known data on thresholds in
Table~\ref{tab:previous} for comparison. The computational resources
rewuired by our computations are discussed in
Appendix~\ref{sec:technicalities}.

\begin{table*}
  \centering
  \begin{tabular}{c|lllllll}
    \diagbox{$P$}{$Q$} & 3 & 4 & 5 & 6 & 7 & 8 & 9 \\\hline
    3 &  & & & \textbf{0.50000000\ldots}
       & 0.26931171(7) & 0.20878618(9) & 0.17157685(3) \\
    4 & \multicolumn{1}{c}{$\pcsite$} & \textbf{0.59274605\ldots} & 0.29890539(6) &  0.22330172(3)
       & 0.17979594(1) & 0.151035321(9) &  0.13045673(2) \\
    5 & & 0.3714769(1) & 0.26186660(5) & 0.20498805(2)
       & 0.16914045(2) & 0.144225373(6) & 0.125818563(3) \\
    6 & \textbf{0.69704023\ldots} & 0.34601617(5)  &  0.25328042(1)  &  0.20115330(2)  &  0.167154812(2)  &  0.143091873(7)  &  0.125124021(3) \\
    7 & 0.54710885(10)  &  0.33788595(1)  &  0.25093250(2)  &  0.200268034(4)  &  0.166762201(2)  &  0.142896751(1)  &  0.1250183755(6) \\
    8 & 0.5221297(4)  &  0.33500594(4)  &  0.250264873(6)  &  0.200061544(2)  &  0.166685043(2)  &  0.1428636871(5)  &  0.1250026592(2) \\
    9 & 0.51118943(5)  &  0.33395047(2)  &  0.2500745527(7)  &  0.2000139338(3)  &  0.1666701374(4)  &  0.1428582040(2)  &  0.1250003781(1)\\
    $\infty$ & 0.50000000\ldots & 0.33333333\ldots & 0.25000000\ldots &
        0.20000000\ldots & 0.16666666\ldots &
       0.14285714\ldots & 0.12500000\ldots
  \end{tabular} \\[1ex]
   \begin{tabular}{c|lllllll}
    \diagbox{$P$}{$Q$} & 3 & 4 & 5 & 6 & 7 & 8 & 9 \\\hline
    3 &  & & & \textbf{0.50000000\ldots}
       &  0.73068829(7)  &  0.79121382(9)  &  0.82842315(3)\\
    4 & \multicolumn{1}{c}{$\pusite$} & \textbf{0.59274605\ldots} &  0.8266384(5)  &  0.87290362(7)  &  0.89897645(3)  &  0.91607962(7)  &  0.92820312(2)\\
    5 & & 0.8500010(2)  &  0.89883342(7)  &  0.9226118(1)  &  0.93725391(5)  &  0.94722182(5)  &  0.95445115(2)\\
    6 & \textbf{0.69704023\ldots} & 0.8980195(2)  &  0.92817467(4)  &  0.94427121(6)  &  0.95445118(6)  &  0.96148136(1)  &  0.96662953(1)\\
    7 & 0.8550371(5)  &  0.9222771(1)  &  0.94426351(9)  &  0.95643895(6)  &  0.96424001(1)  &  0.96966910(1)  &  0.97366600(2)\\
    8 & 0.8911842(4)  &  0.9371043(1)  &  0.95444794(1)  &  0.96424002(1)  &  0.970562733(8)  &  0.97498433(2)  &  0.978250607(6)\\
    9 & 0.9119080(1)  &  0.94714549(5)  &  0.96147998(4)  &  0.969669063(10)  &  0.97498439(2)  &  0.978713769(9)  &  0.981475139(7)\\
    $\infty$ & 1.00000000\ldots & 1.00000000\ldots & 1.00000000\ldots &
        1.00000000\ldots & 1.00000000\ldots &
       1.00000000\ldots & 1.00000000\ldots
  \end{tabular}\\[1ex]
    \begin{tabular}{c|lllllll}
    \diagbox{$P$}{$Q$} & 3 & 4 & 5 & 6 & 7 & 8 & 9 \\\hline
    3 &  & & & \textbf{0.34729635\ldots}
       & 0.1993505(5)  &  0.1601555(2)  &  0.1355650(1) \\
    4 & \multicolumn{1}{c}{$\pcbond$} & \textbf{0.50000000\ldots} &
     0.2689195(3)  &  0.20714787(9)  &  0.17004767(3)  &  0.14467876(3)  &  0.12607213(1)  \\
    5 & & 0.3512228(3)  &  0.25416087(3)  &  0.20141756(5)  &  0.16725887(2)  &  0.143140108(5)  &  0.125148983(6)\\
    6 & \textbf{0.65270364\ldots} & 0.3389049(2)  &  0.25109739(4)  &  0.20031239(1)  &  0.166777706(8)  &  0.142903142(2)  &  0.125021331(2) \\
    7 & 0.5305246(8)  &  0.33526580(4)  &  0.25030153(2)  &  0.20006995(1)  &  0.166687541(3)  &  0.1428645814(8)  &  0.1250030259(6)\\
    8 & 0.5136441(4)  &  0.33402630(3)  &  0.250083308(5)  &  0.200015586(4)  &  0.166670553(1)  &  0.1428583305(8)  &  0.1250004231(2)\\
    9 & 0.5067092(1)  &  0.33358404(2)  &  0.250022914(1)  &  0.200003441(2)  &  0.1666673834(3)  &  0.1428573314(2)  &  0.12500005836(6)\\
      $\infty$ & 0.50000000\ldots & 0.33333333\ldots & 0.25000000\ldots &
        0.20000000\ldots & 0.16666666\ldots &
       0.14285714\ldots & 0.12500000\ldots
  \end{tabular}
  \caption{Percolation thresholds for hyperbolic lattices $\{P,Q\}$. Values
    for the Euclidean lattices (bold) are added for comparison. Note
    that $\pusite = 1-\pcsite$ for triangular lattices ($P=3$), and that $\pubond(\{P,Q\}) = 1-\pcbond(\{Q,P\})$. The
    row $P=\infty$ contains the values $p_c=1/(Q-1)$ and $p_u=1$ for
    the Bethe lattice.  }
  \label{tab:thresholds}
\end{table*}

\begin{table}
  \centering
  \begin{tabular}{c|lll}
     $\{P,Q\}$ & $\pcbond$ & $\pubond$ & Ref.\\\hline
     $\{4,5\}$ & 0.27 & 0.52 & \cite{baek:minhagen:kim:09} \\
     $\{3,7\}$ & 0.20 & 0.37 & \cite{baek:minhagen:kim:09} \\
     $\{7,3\}$ & 0.53 & 0.72 & \cite{baek:minhagen:kim:09} \\
                    & 0.551(10) & 0.810(10) & \cite{gu:ziff:12} \\
    $\{5,5\}$ & 0.263(10) & 0.749(10) & \cite{gu:ziff:12} 
  \end{tabular}
\caption{Previous results for percolation thresholds on hyperbolic lattices. Note that the claimed results for $\{7,3\}$ and $\{3,7\}$ of~\cite{baek:minhagen:kim:09} violate the identity $\pubond(\{P,Q\}) = 1-\pcbond(\{Q,P\})$.}
  \label{tab:previous}
\end{table}
 
These results allow us to explore the dependence of percolation thresholds on $P$ and $Q$.  
Since the shortest loop in a $\{P,Q\}$-lattice has length $P$, the lattice becomes more and more treelike as $P$ gets
larger, and approaches a Bethe lattice as $P \to \infty$.  As mentioned above, the critical densities
for site and bond percolation on a Bethe lattice are known exactly,
\begin{equation}
  \label{eq:pc-pu-bethe}
  p_c(\{\infty,Q\}) = \frac{1}{Q-1} 
  \quad \text{and} \quad 
  p_u(\{\infty,Q\}) = 1 \, .
\end{equation}
We can see this convergence in Table~\ref{tab:thresholds}, where the
rows for $P=9$ and $P = \infty$ are almost identical
for $\pcsite$ and $\pcbond$.  

There is another, more subtle sense in which $\{P,Q\}$ becomes treelike in the limit $Q \to \infty$ with $P$ held fixed: while there exist short loops, the fraction of paths in the graph that complete a loop tends to zero, again suggesting that $p_c$ and $p_u$ should converge to their values on the Bethe lattice.  We can see this in Table~\ref{tab:thresholds} where, even for $P=3$ and $P=4$, both $\pcsite$ and $\pcbond$ quickly converge to $1/(Q-1)$ as $Q$ grows.

This raises the interesting question of how these thresholds approach their values on the Bethe lattice as $P$ or $Q$ increases.  To learn more about this dependence, we start with some rigorous bounds.  The following two theorems 
provide an upper and a lower bound on the site and bond percolation thresholds.
\begin{theorem}
\label{thm:pc-bounds}
Let $p_c(\{P,Q\})$ denote the percolation threshold for either site or bond percolation on the hyperbolic lattice $\{P,Q\}$. Then
  \begin{equation}
    \label{eq:pc-bounds}
      \frac{1}{z_{P,Q}} \le p_c(\{P,Q\}) \le \frac{1}{\lambda_{P,Q}} \, ,
  \end{equation}
  where $z_{P,Q}$ is the largest real root of the polynomial
    \begin{equation}
  \label{eq:char-poly-SAW}
 f_{P,Q} = z^P-(Q-1)z^{P-1}+z+(Q-3) \, .
  \end{equation}
and $\lambda_{P,Q}$ is the largest real root of the polynomial $R_{P,Q}$ given in~\eqref{eq:pc-upper-char-poly}.
\end{theorem}

\begin{proof}
A classic result of Hammersley~\cite{hammersley:57} shows that 
$p_c \ge 1/\mu$ where $\mu$ is the connective constant of the lattice, 
i.e., the exponential growth rate of the number of self-avoiding walks.  
The connective constant for hyperbolic lattices is not known analytically, 
but in Appendix~\ref{sec:saw} we show that $\mu \le z_{P,Q}$ by counting 
paths that do not complete a loop around a face of the lattice.

The upper bound comes from the fact that percolation on a
subgraph $H\subseteq \{P,Q\}$ implies percolation on $\{P,Q\}$,
\begin{equation}
  \label{eq:subgraph}
  p_c(\{P,Q\}) \le p_c(H) \, .
\end{equation}
For $H$ we choose the breadth-first search (BFS) tree of
$\{P,Q\}$. The bond and site percolation thresholds of a tree equal
the reciprocal of its branching ratio. For the BFS, the branching
ratio is given by $\lim_{k \to \infty} n(k+1)/n(k)$, where $n(k)$ denotes
the number of vertices in the hyperbolic lattice at graph distance $k$
from the origin. In Appendix~\ref{sec:counting} we show that $n(k)$ is given by the 
linear recurrence~\eqref{eq:n-p-even-m-even}, \eqref{eq:n-p-even-m-odd},
\eqref{eq:n-p-odd} or~\eqref{eq:n-p-3}, depending on the value of $P$. 
Eqs.~\eqref{eq:pc-upper-char-poly} are the characteristic polynomials of
the recurrences.  The branching ratio $\lambda_{P,Q}$ is the
largest real root of the characteristic polynomial.
\end{proof}

What does Theorem~\ref{thm:pc-bounds} tell us about the how $p_c$ approaches $1/(Q-1)$?  For large $P$ or large $Q$, the largest root of $f_{P,Q}$ (indeed, the unique positive root) converges to $Q-1$. If we plug in the ansatz $z_{P,Q} = (Q-1)-\eps$ and expand $f_{P,Q}$ to linear order in $\eps$, we get
\begin{equation}
  \label{eq:root-SAW}
  z_{P,Q} = Q - 1 - \frac{2(Q-2)}{(Q-1)^{P-1}} + O\big( (Q-1)^{-(2P-3)} \big) \, ,
\end{equation} 
so the lower bound in~\eqref{eq:pc-bounds} scales as 
\begin{equation}
\label{eq:pc-lower-bound}
p_c \ge \frac{1}{Q-1} + O\left( \frac{1}{(Q-1)^P} \right) \, .
\end{equation}
Similarly, for $P > 4$ even~\eqref{eq:pc-upper-char-poly-P-even-m-even}, \eqref{eq:pc-upper-char-poly-P-even-m-odd}, we have
\begin{equation}
\label{eq:bound-P-even}
\lambda_{P,Q} = Q - 1 - \frac{Q(Q-2)}{(Q-1)^{P/2}} + O\left( (Q-1)^{-(P-2)} \right) \, , 
\end{equation}
and for $P > 3$ odd~\eqref{eq:pc-upper-char-poly-P-odd} we have
\begin{equation}
\label{eq:bound-P-odd}
\lambda_{P,Q} = Q - 1 - \frac{2(Q-2)}{(Q-1)^{(P-1)/2}} + O\left( (Q-1)^{-(P-3)} \right) \, .
\end{equation}
Thus the upper bound in~\eqref{eq:pc-bounds} scales as
\begin{equation}
\label{eq:pc-upper-bound}
p_c \le \frac{1}{Q-1} + O\left( \frac{1}{(Q-1)^{\lceil P/2 \rceil}} \right) \, .
\end{equation}

Combining~\eqref{eq:pc-lower-bound} and~\eqref{eq:pc-upper-bound} 
suggests that the critical density of the hyberbolic lattice should scale with $P$ as
\begin{equation}
  \label{eq:pc-scaling}
  p_c(\{P,Q\}) = \frac{1}{Q-1} + \frac{c}{\mu_Q^P} 
\end{equation}
for some constant $c$ and some $\sqrt{Q-1} \le \mu_Q \le Q-1$.  In Fig.~\ref{fig:pcsite-scaling} we plot $\pcsite(\{P,Q\})-1/(Q-1)$ versus $P$ on a semilog scale, and the straight lines indicate that scaling of the form~\eqref{eq:pc-scaling} holds.  We also compare $\pcsite-1/(Q-1)$ for $Q=7$ directly to our rigorous bounds in Fig.~\ref{fig:pc-bounds}.

We can look more closely at the cases $P=3$ and $P=4
$ when $Q$ is large.  If $P=3$~\eqref{eq:pc-upper-char-poly-P-3}, we have
\begin{align}
\lambda_{P,Q} 
&= \frac{Q + \sqrt{(Q-6) (Q-2)} - 4}{2} \nonumber\\
&= Q - 4 - \frac{1}{Q-4} + O(1/Q^{3}) \, , \label{eq:bound-P-3}
\end{align}
and when $P=4$~\eqref{eq:pc-upper-char-poly-P-4} we have
\begin{align}
\lambda_{P,Q} 
&= \frac{Q + \sqrt{Q(Q-4)} - 2}{2} \nonumber \\
&= Q - 2 - \frac{1}{Q-2} + O(1/Q^3) \, . \label{eq:bound-P-4}
\end{align}
We can explain~\eqref{eq:bound-P-3} and~\eqref{eq:bound-P-4} by thinking about breadth-first search on these lattices.  If $P=3$, then each vertex $v$ in the BFS tree has at least one edge pointing back to its parent, and two edges pointing laterally to other vertices $u,w$ at the same level.  Moreover, the triangle below each of those edges has a corner on the next level which can be reached either from $v$ or from its other ``parent'' $u$ or $w$.  This reduces the branching ratio to at most $Q-4$.  

Similarly, for $P=4$ most vertices $v$ in the BFS tree belong to a square face where $v$ is the farthest from the root, and $v$'s neighbors on that face are both closer to the origin.  This gives $v$ two ``parents'' in the previous layer, and reduces the branching ratio to $Q-2$.  

On the other hand, for $P \ge 5$ the bound of~\eqref{eq:bound-P-odd} converges to $Q-1$ for large $Q$.  This indicates that for most vertices, only one of their neighbors is in the previous layer, and the other $Q-1$ are in the succeeding layer.

\begin{figure}
  \centering
  \includegraphics[width=\columnwidth]{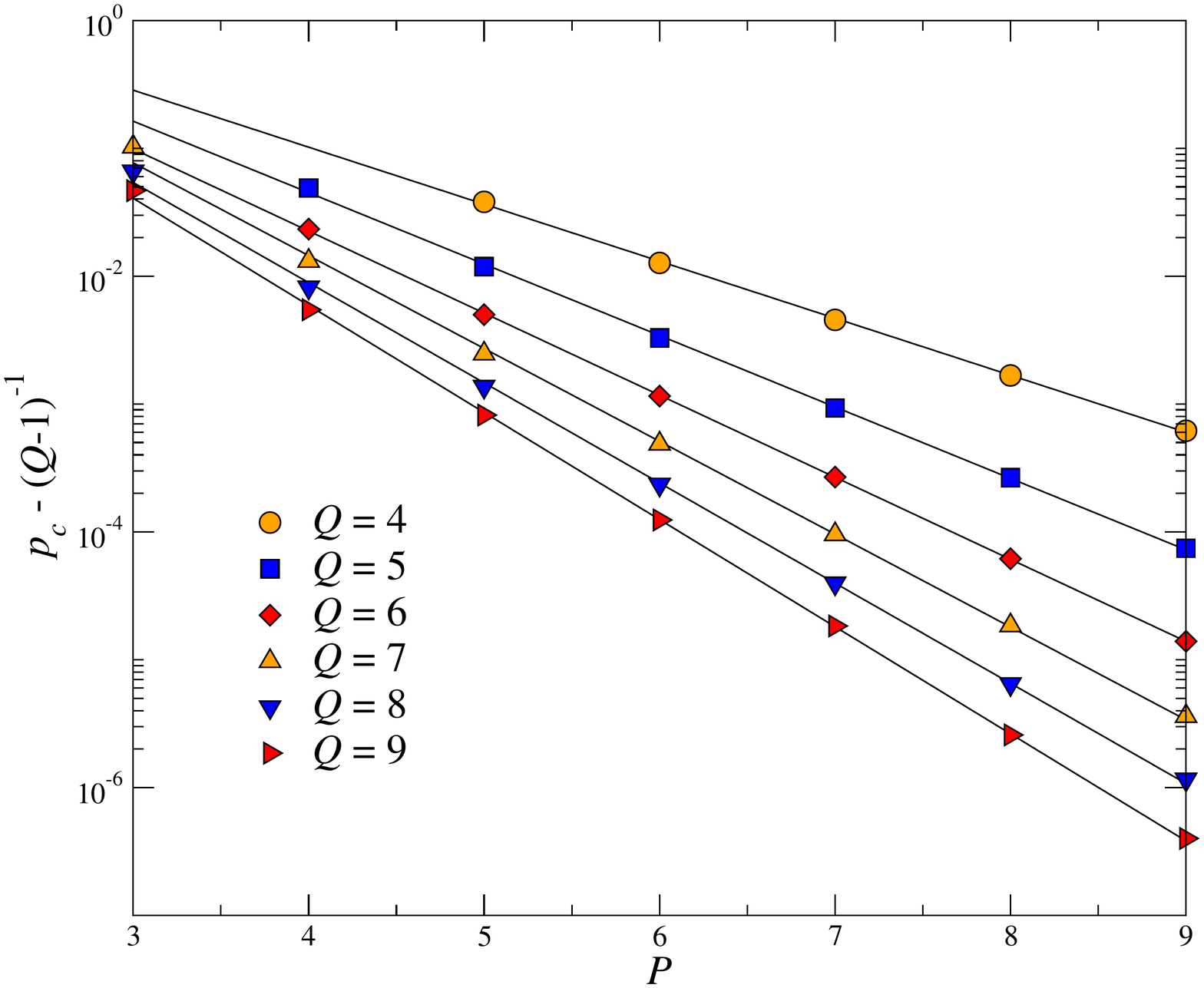}
  \vspace{-24pt}
  \caption{Scaling of $\pcsite(\{P,Q\})$ with $P$.}
  \label{fig:pcsite-scaling}
\end{figure}

\begin{figure}
  \centering
  \includegraphics[width=\columnwidth]{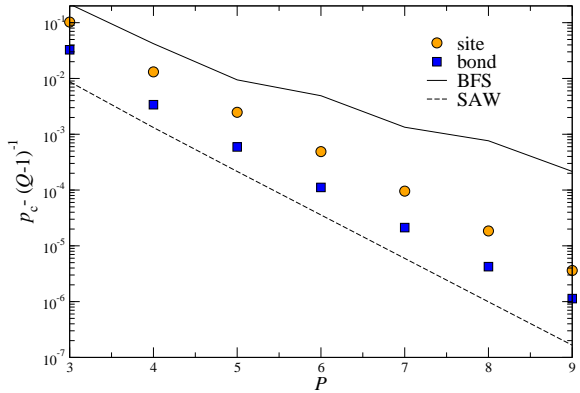}
  \vspace{-24pt}
  \caption{Comparison of $\pcsite(\{P,Q\})$ for $Q=7$ with the upper and lower bounds from Theorem~\ref{thm:pc-bounds}.}
  \label{fig:pc-bounds}
\end{figure}

We also prove the following simple bounds on the uniqueness threshold (see also~\cite{Delfosse:2013,lee-baek:2012} for some combinatorial bounds for $\{5,5\}$, $\{5,4\}$, and $\{4,5\}$).

\begin{theorem}
  \label{thm:pu-bounds}
  Let $\pusite$ be the uniqueness threshold for site percolation on
  the hyperbolic lattice $\{P,Q\}$, and let 
  \[
  a = (P-2)(Q-2)-2 \, . 
  \] 
  Then
  \begin{equation}
     \label{eq:pu-bounds}
     1-\frac{2}{a + \sqrt{a^2-4}} \le \pusite \le 1-\frac{1}{Q(P-2)-1} \, .
  \end{equation}
\end{theorem}

\noindent 
Note that when $a$ is large (i.e., if $P$ or $Q$ is large) the lower bound in~\eqref{eq:pu-bounds} tends to $1-1/a-1/a^3$.

\begin{proof}
According to~\eqref{eq:site-duality}, an upper (lower) bound for $\pusite$ is equivalent to a lower (upper) bound for $\pmatch$, the site percolation threshold on the matching lattice. We will prove bounds for $\pmatch$ following the same ideas as in Theorem~\ref{thm:pc-bounds}.

An upper bound for $\pmatch$ is given by the reciprocal of the branching ratio of the BFS tree of the matching lattice. Since all the vertices of a face of the underlying hyperbolic lattice are adjacent in the matching lattice, and since each face has a constant number of vertices, this branching ratio is also given by
\[
    \lim_{\ell \to \infty} \frac{n(\ell+1)}{n(\ell)} \, ,
\]
where $n(\ell)$ denotes the number of vertices on the perimeter of the $\ell$th layer of faces around the origin. In Appendix~\ref{sec:counting} we use a linear recurrence to compute $n(\ell)$ analytically, and~\eqref{eq:binet} shows that the branching ratio is $(a+\sqrt{a^2-4})/2$.  This gives the lower bound on $\pusite$ in~\eqref{eq:pu-bounds}.

The upper bound in~\eqref{eq:pu-bounds} follows from the fact that the connective constant of a $d$-regular graph is at most $d-1$. The matching lattice is regular with $d=Q(P-2)$.  
\end{proof}

\begin{figure}
  \centering
  \includegraphics[width=\columnwidth]{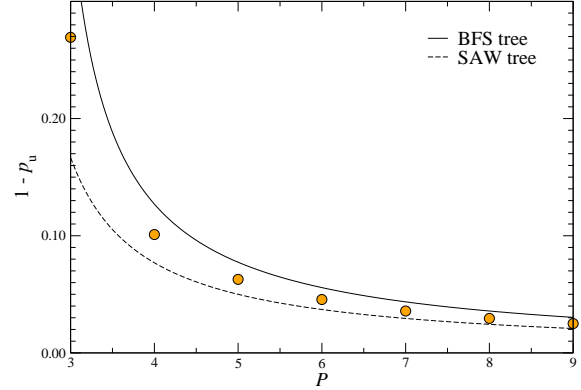}
  \vspace{-24pt}
  \caption{Uniqueness threshold for $Q=7$ and $P=3,\ldots,9$ and the bounds from Theorem~\ref{thm:pu-bounds}.}
  \label{fig:pu-bounds}
\end{figure}

Fig.~\ref{fig:pu-bounds} compares our measurements of the uniqueness threshold with the bounds of Theorem~\ref{thm:pu-bounds}.  We could easily tighten these bounds, for instance by bounding the connective constant on the matching lattice by prohibiting short loops as in Theorem~\ref{thm:pc-bounds}.

Finally, recall that $\pubond$ for $\{P,Q\}$ is simply $1-\pcbond$ on the dual lattice $\{Q,P\}$.  Thus we can read bounds on $\pubond$ directly from Theorem~\ref{thm:pc-bounds}. 

\section{Mean Distance}
\label{sec:mean-distance}

Another quantity of interest is the mean graph distance $\langle k_N\rangle$ of a vertex from the origin of the cluster of mass $N$. This quantity can be computed exactly on a tree, for the invasion percolation cluster as well as for the incipient infinite cluster of Bernoulli percolation~\cite{nickel:wilkinson:83}.  The exact expressions for finite $N$ involve hypergeometric functions, but here we care only about the asymptotic behavior. For Bernoulli percolation on a $Q$-regular tree, the asymptotic expression is
\begin{equation}
  \label{eq:mean-distance-tree-bernoulli}
  \langle k_N\rangle = \sqrt{\frac{\pi(Q-1)}{2(Q-2)}}\,N^{\frac{1}{2}}
  - \frac{4(Q-1)+1}{3(Q-2)} + o(1)\,.
\end{equation}
For invasion percolation on a tree, the leading term also grows as $N^{\frac{1}{2}}$, but the next term grows logarithmically:
\begin{equation}
    \label{eq:mean-distance-tree-invasion}
    \langle k_N\rangle = \frac{4}{3}
    \sqrt{\frac{\pi(Q-1)}{2(Q-2)}}\,N^{\frac{1}{2}} -\frac{1}{3}
    \frac{Q}{Q-2}\,\ln N + O(1)
\end{equation}
We claim, that on hyperbolic lattices, the mean distance of a vertex of the invasion cluster, has scaling similar to~\eqref{eq:mean-distance-tree-invasion}, i.e.,
\begin{equation}
    \label{eq:mean-distance}
    \langle k_N\rangle = c_1(P,Q)\,N^{\frac{1}{2}} - c_2(P,Q)\,\ln N + O(1)\,.
\end{equation}
This claim is clearly supported by our data, see Fig.~\ref{fig:kappa}.
The constants $c_1$ and $c_2$ differ from the corresponding values in
\eqref{eq:mean-distance-tree-invasion}, but they both converge to
these values as $P$ gets large.

\begin{figure}
  \centering
  \includegraphics[width=\columnwidth]{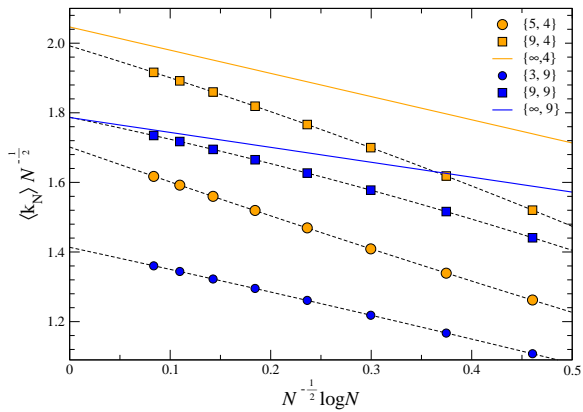}
  \vspace{-24pt}
  \caption{Mean distance $\langle k_N\rangle$ of a vertex of the invasion cluster of mass $N$ to the origin. The lines $\{\infty,Q\}$ represent the asymptotics~\eqref{eq:mean-distance-tree-invasion} for the $Q$-regular tree. The data shown is for bond percolation, and the dotted lines are quadratic fits.}
  \label{fig:kappa}
\end{figure}

We also measured the maximum distance of a vertex on the
invasion cluster from the origin. This maximum distance also scales
like \eqref{eq:mean-distance}, with different values for $c_1$ and
$c_2$, of course.  

\section{Critical Exponents}
\label{sec:exponents}

It has been rigorously established that the critical exponents of percolation in hyperbolic lattices are equal to their mean-field values, i.e., their values on the Bethe lattice. This first proof by Schonmann~\cite{schonmann:02} was based on planarity and non-amenability, but later Madras and Wu~\cite{madras:wu:10} showed that non-amenability suffices. Thus any negative curvature puts the lattice in the universality class of the tree. The same holds for the critical exponents of the Ising model on hyperbolic lattices~\cite{shima:sakaniwa:06,shima:sakaniwa:06a,krcmar:2008,gendiar:2012,serina:2016}.

Since the critical exponents for percolation on hyperbolic lattices are known with mathematical rigor, there is no need to measure them numerically.  However, to check that our numerics are consistent with these rigorous results, we compute the Fisher exponent $\tau$.  

According to scaling theory, the average number of clusters of size
$s$ per lattice site scales as
\begin{equation}
  \label{eq:ns-scaling}
  n_s(p) = s^{-\tau}\,[f_0(z) + s^{-\Omega} f_1(z) + \cdots] \, ,
\end{equation}
where the scaling functions $f_0$ and $f_1$ are analytic for small
values of $z$, and where the scaling variable $z$ is defined as
\begin{equation}
  \label{eq:z-scaling}
  z = (p-p_c)s^\sigma \, .
\end{equation}
The exponent $\tau$ can be measured by growing single clusters at
$p=p_c$ and recording their sizes; this is known as the Leath
algorithm~\cite{leath:76a,leath:76b}.  This is similar to invasion
percolation except that it adds every neighboring vertex with weight $p \le p_c$ to the cluster instead of just the vertex with minimum weight.

\begin{figure}
  \centering
  \includegraphics[width=\columnwidth]{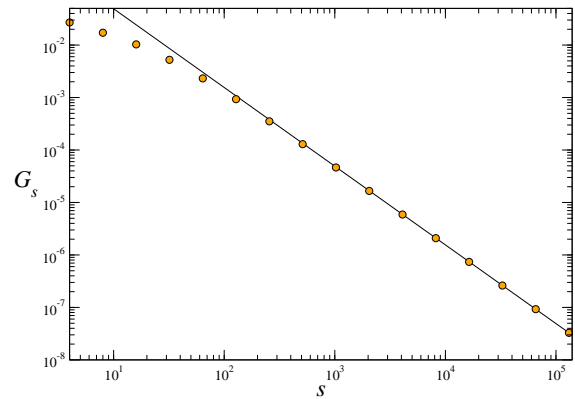}
  \vspace{-24pt}
  \caption{Scaling of the cluster density at criticality for the $\{7,3\}$ lattice using logarithmic binning. The line is
    $G_s \sim s^{1-\tau}$ for $\tau=5/2$, confirming that $\tau$ equals its mean-field value.}
  \label{fig:tau-7-3}
\end{figure}

In order to improve the count statistics on the larger (and rarer) clusters, the individual $n_s$ are binned on a logarithmic scale~\cite{rapaport:86}.  Thus the actual quantity we analyze is, where $s = 2^k$ for integer $k$, 
\begin{equation}
  \label{eq:def-Gs}
  G_s = \sum_{s'=s}^{2s-1} n_{s'}(p_c) \, .
\end{equation}
This scales as $\sim s^{1-\tau}$ for large values of $s$.  We show this for site percolation on the $\{7,3\}$ lattice in
Fig.~\ref{fig:tau-7-3}.  Our numerical results show that $\tau$ does indeed equal its mean-field value $\tau=5/2$, and results for other $P,Q$ are similar.

\section{Conclusions}

We have implemented the invasion percolation algorithm on the hyperbolic lattice, using a new combinatorial labeling of the vertices to make this algorithm computationally feasible.  This yields highly-accurate measurements of the thresholds for site and bond percolation, as well as (by using the matching lattice in the case of site percolation) of the threshold at which the infinite cluster becomes unique.  These measurements, in turn, allow us to compare these thresholds with rigorous upper and lower bounds, and to confirm experimentally that the critical exponents of percolation in negatively curved spaces are equal to their mean-field values.

\appendix

\section{Counting vertices in hyperbolic lattices}
\label{sec:counting}

The number of vertices a given distance from a specified site (the ``origin'') is given by a linear recurrence. This recurrence is simple if we measure the distance by the number $\ell$ of layers of faces that separate the vertices from the origin.  We will show that the number of vertices on the perimeter of the $\ell$th layer is given by
\begin{equation}
  \label{eq:n-ell-recursion}
  \begin{aligned}
    n(0) &= 0 \\
    n(1) &= (P-2)Q \\
  n(\ell) &= a n(\ell-1) - n(\ell-2) \qquad \ell > 1 \, ,
  \end{aligned}
\end{equation}
with
\begin{equation}
\label{eq:a}
a = (P-2)(Q-2)-2 \, .
\end{equation}
Like any homogeneous linear recurrence with constant coefficients, \eqref{eq:n-ell-recursion} can be
solved to give
\begin{equation}
  \label{eq:binet}
  n(\ell) =
  \frac{(P-2)Q}{\sqrt{a^2-4}}\left[\left(\frac{a+\sqrt{a^2-4}}{2}\right)^\ell
    - \left(\frac{a-\sqrt{a^2-4}}{2}\right)^\ell\right] \, .
\end{equation}
For hyperbolic lattices we have $a > 2$, and $n(\ell)$ grows exponentially with branching ratio $(a+\sqrt{a^2-4})/2$.  For Euclidean lattices, $a = 2$ and~\eqref{eq:binet} reduces to linear growth $n(\ell)=(P-2)Q\ell$.

To prove~\eqref{eq:n-ell-recursion} we start with the base case of the recurrence.  The first layer of faces consists of $Q$ polygons with $P$ vertices each. In the total count $PQ$ of vertices, the origin is counted $Q$ times, and in the perimeter, the $Q$ vertices connected to the origin are shared by adjacent polygons and counted twice. Hence $n(1)=(P-2)Q$.

Now let $F(\ell)$ denote the number of faces in layer $\ell$.  Each face has $P$ edges, two of them crossing from the inner boundary of this layer to its outer boundary. 
The remaining $(P-2) F(\ell)$ edges belong to the inner or outer boundary of the layer. Since each of these boundaries forms a cycle, the number of vertices on the boundary equals the number of edges. Hence we have
\begin{equation}
  \label{eq:proof-planarity}
  n(\ell)+n(\ell-1) = (P-2) F(\ell) \, .
\end{equation}
On the other hand, each vertex has $Q-2$ edges that point either inward or outward across a layer, and each of these edges corresponds uniquely to one polygon in layer $\ell$ or $\ell-1$. Hence we also have
\begin{equation}
  \label{eq:proof-dual}
  F(\ell) + F(\ell-1) = (Q-2) n(\ell-1) \, .
\end{equation}
Note the duality of~\eqref{eq:proof-planarity} and~\eqref{eq:proof-dual}.
Now adding~\eqref{eq:proof-planarity} to itself with
$\ell\mapsto\ell-1$ gives us
\begin{align*}
  n(\ell)+2n(\ell-1)+n(\ell-2) &= (P-2)\,\big(F(\ell)+F(\ell-1)\big)\\
  &= (P-2)(Q-2) n(\ell-1) \, ,
  \end{align*}
yielding~\eqref{eq:n-ell-recursion} and~\eqref{eq:a}.  Note that because of the linear relation between $F(\ell)$ and $n(\ell)$, the number of faces $F(\ell)$ obeys the the same recurrence but with base case $F(0)=0$ and $F(1)=q$.

Another measure of distance is given by the graph distance, i.e. by the length of the shortest path that connects two vertices.  Let $n_k$ denote the number of vertices in the hyperbolic lattice with graph distance $k$ from the origin.  This number is again given by a linear recurrence, albeit a more complicated one, which was derived independently in physics~\cite{okeeffe:98} and in mathematics~\cite{paul:pippenger:11}.  We review this here, using the notation of~\cite{okeeffe:98}.

The recurrence depends on whether $P$ is even or odd, and in the even case on $P \bmod 4$.  For $P=2m$ where $m$ is even, it reads
\begin{equation}
  \label{eq:n-p-even-m-even}
  n_{k+1} = (Q-2) \sum_{i=0}^{m-2} n_{k-i} - n_{k-m+1} \, , 
\end{equation}
while if $m$ is odd we have
\begin{equation}
  \label{eq:n-p-even-m-odd}
  n_{k+1} = \sum_{i=0}^{(m-3)/2} \Big[(Q-1)n_{k-2i} - n_{k-2i-1}\Big] \, . 
\end{equation}
The initial values are
\begin{equation}
  \label{eq:n-p-even-initials}
  n_k = \begin{cases}
     0 & \quad k \le 0 \\
   (Q-1)^{k-1} Q & \quad 0 < k < m \, .
  \end{cases}
\end{equation}
For odd values $P=2m+1$ and $m > 1$ the recurrence is
\begin{equation}
  \label{eq:n-p-odd}
  n_{k+1} = (Q-2) \sum_{i=0}^{m-2} (n_{k-i}+n_{k-m-i})  + (Q-4) n_{k-m+1} -n_{k-2m+1}
\end{equation}
with initial values
\begin{equation}
  \label{eq:n-p-odd-initials}
  n_k = (Q-1)^{k-1} Q \qquad 0 < k \le m \, .
\end{equation}
For $P=3$, the sum in~\eqref{eq:n-p-odd} disappears, and the recurrence becomes
\begin{equation}
  \label{eq:n-p-3}
  n_{k+1} = (Q-4) n_k - n_{k-1}
\end{equation}
with initial values
\begin{equation}
  \label{eq:n-p-3-initials}
  n_{k} = \begin{cases}
   0 & \quad k \le 0 \\
   Q & \quad k = 1 \, .
 \end{cases}
\end{equation}
The corresponding characteristic polynomials are as follows.  
For $P=2m$ and $m$ even (i.e., $P \bmod 4=0$), 
\begin{subequations}
   \label{eq:pc-upper-char-poly}
   \begin{equation}
    \label{eq:pc-upper-char-poly-P-even-m-even}
    R_{P,Q}(z) = z^m \left( 1 - \frac{Q-2}{z-1} \right) + \frac{Q-2}{z-1} + Q-1 \, .
   \end{equation}
For $P=2m$ and $m$ odd (i.e., $P \bmod 4=2$), 
   \begin{equation}
    \label{eq:pc-upper-char-poly-P-even-m-odd}
    R_{P,Q}(z) = 
    z^m \left( \frac{z-(Q-1)}{z^2-1} \right) + \frac{z(Q-1)-1}{z^2-1} \, . 
   \end{equation}
For $P=2m+1$ (i.e., $P$ is odd),
  \begin{equation}
    \label{eq:pc-upper-char-poly-P-odd}
    R_{P,Q}(z) = 
      z^{2m} \left( 1-\frac{Q-2}{z-1} \right) + 2z^m+\frac{z(Q-1)-1}{z-1}
   \end{equation}
In particular, for $P=3$ ($m=1$) we have 
    \begin{equation}
    \label{eq:pc-upper-char-poly-P-3}
    R_{3,Q}(z) = z^2 - (Q-4) z + 1 \, ,
   \end{equation}
and for $P=4$, setting $m=2$ in~\eqref{eq:pc-upper-char-poly-P-even-m-even} gives
\begin{equation}
    \label{eq:pc-upper-char-poly-P-4}
    R_{4,Q}(z) = z^2 - (Q-2) z + 1 \, .
   \end{equation}
  \end{subequations}

\begin{figure}
  \centering
  \includegraphics[width=\columnwidth]{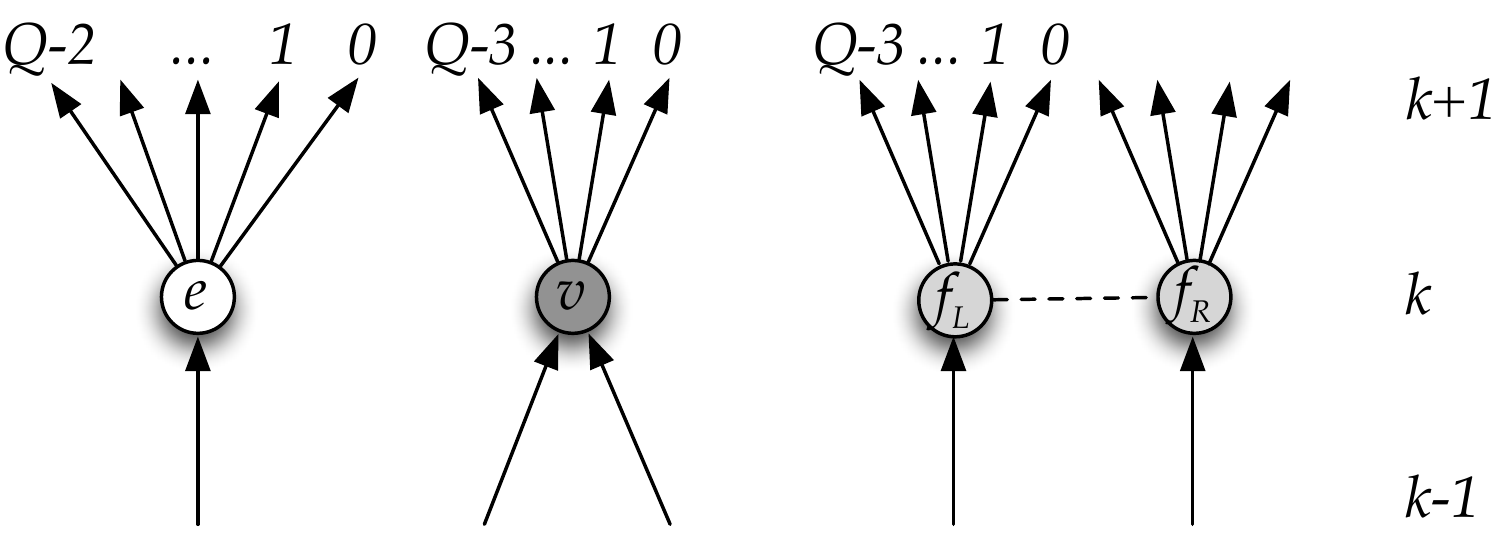}
  \caption{Types of vertices in layer $k$ and their edge labelings.}
  \label{fig:vertices}
\end{figure}

To prove~\eqref{eq:n-p-even-m-even}, \eqref{eq:n-p-even-m-odd}, \eqref{eq:n-p-odd} and~\eqref{eq:n-p-3}, we consider all vertices in the layer at distance $k$. We assign directions to the edges that connect vertices in layer $k$ to vertices in layer $k+1$ (see Fig.~\ref{fig:vertices}). 

For even values of $P$ there are two types of vertices: ``$e$-vertices'' have $1$ incoming edge and $Q-1$ outgoing edges, and ``$v$-vertices'' have $2$ incoming edges and $Q-2$ outgoing edges. For odd values of $P$ there are also pairs of ``$f$-vertices'' in the same layer connected to each other by an undirected edge, and each one has one incoming edge and $Q-2$ outgoing edges. 

Let $x_k$ denote the number of vertices of type $x \in \{e,v,f\}$.  For $P=2m$ we know that each $v$-vertex in layer $k$ ``closes'' a polygon that ``opened'' in layer $k-m$. Each $e$-vertex opens $Q-2$ polygons in the higher layers, and each $v$-vertex opens $Q-3$ polygons. Hence
\begin{equation}
  \label{eq:v-rec-p-even}
  v_{k} = (Q-2)e_{k-m} + (Q-3) v_{k-m} \, .
\end{equation}
The number of edges that connect layer $k-1$ and $k$ is $e_k+2v_k$,
but it is also given by $
(Q-1)e_{k-1} + (Q-2) v_{k-1}$. 
Hence
\begin{equation}
  \label{eq:e-rec-p-even}
  e_{k} = (Q-1)e_{k-1} + (Q-2) v_{k-1}-2v_k \, .
\end{equation}
Finally, the total number of vertices is 
\begin{equation}
  \label{eq:n-rec-p-even}
  n_k = v_k + e_k \, .
\end{equation}
This gives us three equations for $n_k$, $v_k$ and $e_k$, and eliminating $e_k$ and $v_k$ yields~\eqref{eq:n-p-even-m-even} and~\eqref{eq:n-p-even-m-odd}.

For the odd case $P=2m+1$, we again consider the number of edges that connect layer $k-1$ and $k$ to obtain
\begin{equation}
  \label{eq:e-rec-p-odd}
  e_{k} = (Q-1)e_{k-1} + (Q-2) v_{k-1}+(Q-2) f_{k-1}-2v_k \, .
\end{equation}
Each polygon that is closed by a single $v$-vertex in layer $k$ is opened by two $f$-vertices in layer $k-m$, so
\begin{equation}
  \label{eq:v-rec-p-odd}
  v_{k} = \frac{1}{2}\,f_{k-m} \, .
\end{equation}
The number of polygons closed by two $f$-vertices in layer $k$ equals the number of polygons opened by single vertices in layer $k-m$, so
\begin{equation}
  \label{eq:f-rec-p-odd}
  \frac{1}{2}\,f_{k} = (Q-2) e_{k-m} + (Q-3) v_{k-m} + (Q-3) f_{k-m} \, .
\end{equation}
In this case the total number of vertices is
\begin{equation}
  \label{eq:n-rec-p-odd}
  n_k = v_k + e_k+f_k \, ,
\end{equation}
and eliminating $v_k$, $e_k$, and $f_k$ yields~\eqref{eq:n-p-odd}.

\begin{figure}
  \centering
  \includegraphics[width=\columnwidth]{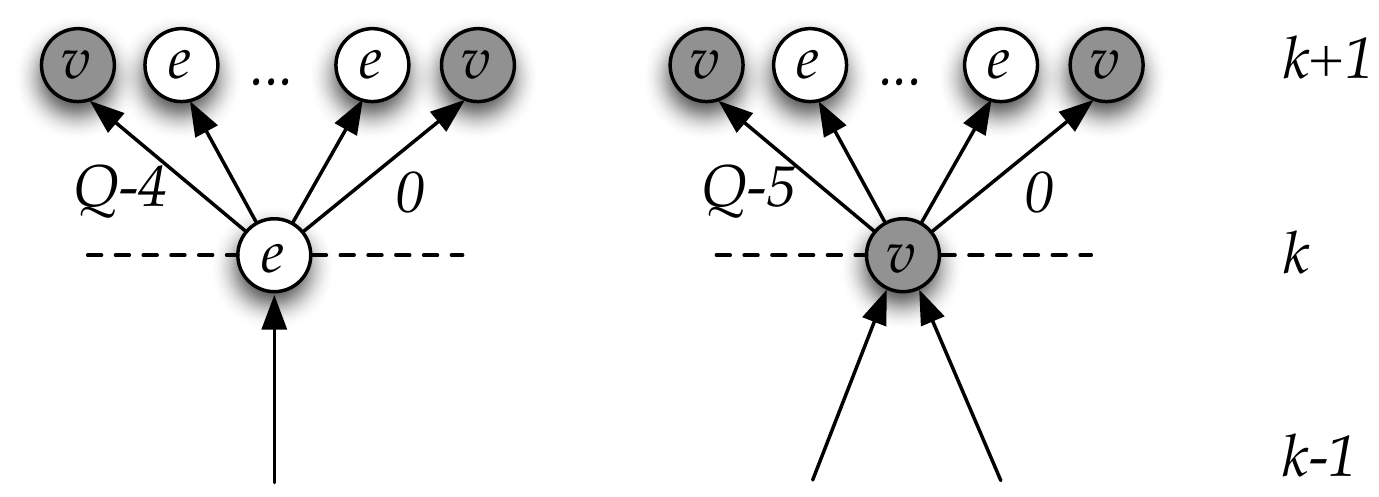}
  \caption{Types of vertices and their edge labelings for $P=3$.}
  \label{fig:trivertices}
\end{figure}

Although we can view it as a degenerate case of~\eqref{eq:n-p-odd}, we end with a specialized derivation of~\eqref{eq:n-p-3}, the recurrence for $P=3$.  In triangular lattices, there are only two types of vertices: $e$-vertices with one incoming edge, and $v$-vertices with two incoming edges.  As shown in Fig.~\ref{fig:trivertices}, every vertex is connected to two vertices in the same layer, leaving a total of $Q-2$ edges to connect to other layers.  Moreover, each vertex in layer $k$ is connected to two $v$-vertices in layer $k+1$. But since each $v$-vertex has two incoming edges,
\[
  v_{k+1} = n_k \, .
\]
Counting the number of outgoing edges to $e$-vertices gives
\[
  e_{k+1} = (Q-5) e_k + (Q-6) v_k \, ,
\]
and combining these equations yields~\eqref{eq:n-p-3}.

\section{Implementing hyperbolic lattices}
\label{sec:implementation}

As discussed above, the challenge for implementing invasion percolation, and indeed any algorithm on hyperbolic lattices, is the lack of a simple coordinate system.  We need a computationally efficient way to index vertices, and to compute the indices of their neighbors.  
Here we describe a ``coordinate system'' that assigns a unique string to each vertex, and we give a procedure for computing the strings in its neighborhood. 
 
The idea is to label each vertex according to one of the shortest paths that connects it to the origin of the lattice.  In essence, we do this by re-doing the derivations of the linear recurrences in Appendix~\ref{sec:counting}, while keeping track of the labels of individual edges in the path. Thus the origin is represented by the empty string, and each vertex in layer $k$ corresponds to a string of length $k$. 

The labeling of the edges is depicted in Figs.~\ref{fig:vertices} and~\ref{fig:trivertices}. Note that we only label outgoing edges, i.e., edges that run between layers $k$ and $k+1$. Edges that connect vertices in the same layer are never part of a shortest path, so there is no need to assign labels to these edges.

The subset of directed (and therefore labeled) edges induces a subgraph that is almost a directed tree. Only the $v$-vertices with their two incoming edges cause loops by closing a face. To break the resulting ties, we never use the right incoming edge of a $v$-vertex.  With this rule, the subgraph induced by the allowed directed edges is a tree, and we denote the unique path from the root to a vertex $u$ in level $k$ as $(u_1, u_2, \ldots, u_k)$.

Suppose $u$ is in layer $k$.  Deleting the last edge in the path to $u$ yields $u$'s ``parent'' in layer $k-1$.  If we extend the path to $u$ by one more edge yields a ``child'' $v$ in layer $k+1$; however, this path might violate the above rule, so the path to $v$ in the tree might not go through $u$. In addition, $u$ may have neighbors in its own layer, which are neither its parents nor its children. The following procedures are useful to perform all this navigation:
\begin{itemize}
\item $\vertextype(u)$ returns the type ($e$, $v$, $f_R$
  or $f_L$) of vertex $u$,
\item $\outdegree(u)$ returns the number of $u$'s outgoing edges,
\item $\child(u,x)$ returns the child of $u$ along an outgoing edge with label $x$,
\item $\parent(u)$ returns the parent of $u$,
\item $\suc(u)$ returns the next vertex in $u$'s layer moving counterclockwise, and
\item $\pred(u)$ returns the next vertex in $u$'s layer moving clockwise.
\end{itemize}
Note that $\suc(u)$ and $\pred(u)$ may or may not be neighbors of $u$.

\begin{figure}

\includegraphics[width=0.6\columnwidth]{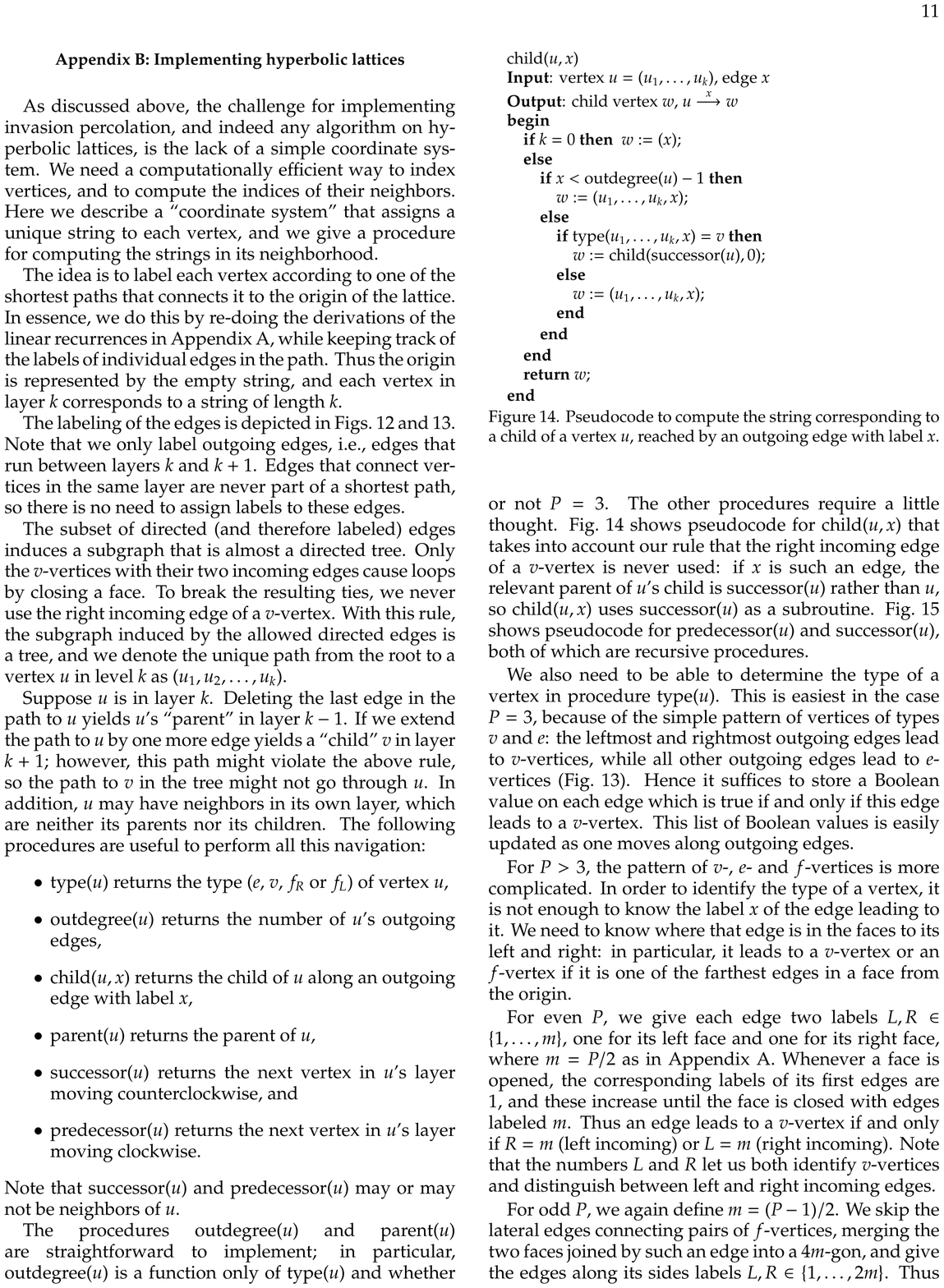}
\vspace{-8pt}
\caption{Pseudocode to compute the string corresponding to a child of a vertex $u$, reached by an outgoing edge with label $x$.}
\label{fig:child}
\end{figure}

The procedures $\outdegree(u)$ and $\parent(u)$ are straightforward to implement; in particular, $\outdegree(u)$ is a function only of $\vertextype(u)$ and whether or not $P=3$. The other procedures require a little thought. Fig.~\ref{fig:child} shows pseudocode for $\child(u,x)$ that takes into account our rule that the right incoming edge of a $v$-vertex is never used: if $x$ is such an edge, the relevant parent of $u$'s child is $\suc(u)$ rather than $u$, so $\child(u,x)$ uses $\suc(u)$ as a subroutine. Fig.~\ref{fig:sucpred} shows pseudocode for $\pred(u)$ and $\suc(u)$, both of which are recursive procedures.

\begin{figure}
 


 

\includegraphics[width=0.83\columnwidth]{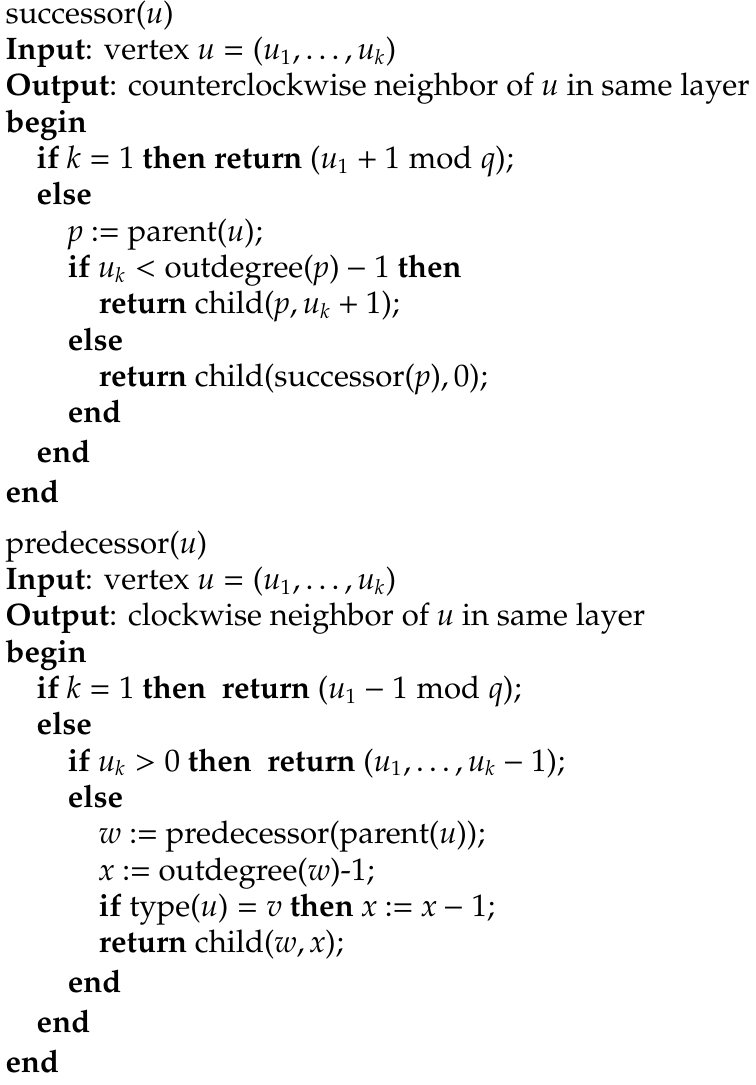}
 \vspace{-8pt}
\caption{Pseudocode to compute the next vertex in the same layer at $u$, moving clockwise or counterclockwise.}
\label{fig:sucpred}
\end{figure}

We also need to be able to determine the type of a vertex in procedure $\vertextype(u)$. This is easiest in the case $P=3$, because of the simple pattern of vertices of types $v$ and $e$: the leftmost and rightmost outgoing edges lead to $v$-vertices, while all other outgoing edges lead to $e$-vertices (Fig.~\ref{fig:trivertices}).  Hence it suffices to store a Boolean value on each edge which is true if and only if this edge leads to a $v$-vertex. This list of Boolean values is easily updated as one moves along outgoing edges.

For $P > 3$, the pattern of $v$-, $e$- and $f$-vertices is more complicated.  The type of a vertex is not determined by the label $x$ of the edge leading to it: we need to know where that edge is in the faces to its left and right. In particular, it leads to a $v$-vertex or an $f$-vertex if it is one of the farthest edges in a face from the origin.  

For even $P$, we give each edge two labels $L, R \in \{1,\ldots,m\}$, one for its left face and one for its right face, where $m = P/2$ as in Appendix~\ref{sec:counting}.  Whenever a face is opened, the corresponding labels of its first edges are $1$, and these increase until the face is closed with edges labeled $m$.  Thus an edge leads to a $v$-vertex if and only if $R=m$ (left incoming) or $L=m$ (right incoming).  Note that the numbers $L$ and $R$ let us both identify $v$-vertices and distinguish between left and right incoming edges.  

For odd $P$, we again define $m= (P-1)/2$.  We skip the lateral edges connecting pairs of $f$-vertices, merging the two faces joined by such an edge into a $4m$-gon, and give the edges along its sides labels $L, R \in \{1,\ldots,2m\}$.  Thus an edge leads to a $v$-vertex if and only if $R=2m$ (left incoming) or $L=2m$ (right incoming).  Edges with $L=m$ or $R=m$ lead to $f_R$- or $f_L$-vertices respectively.

These additional labels can be maintained by the $\child$, $\suc$ and $\pred$ procedures.  This makes the full implementation a bit more complicated than the pseudocode in Figs~\ref{fig:child} and~\ref{fig:sucpred}, but the corresponding code is easily added.  We provide a working implementation in C++ at~\cite{mertens:percolation}.

Equipped with these procedures, the $Q$ neighbors of a vertex $u$ other than the origin can be computed as
\begin{itemize}
\item $\parent(u)$,
\item $\pred(\parent(u))$ if $u$ is a $v$-vertex,
\item $\pred(u)$ if $u$ is an $f_L$-vertex or $P=3$,
\item $\suc(u)$ if $u$ is an $f_R$-vertex or $P=3$, and
\item $\child(u,x)$ for $x=0,\ldots,\outdegree(u)-1$.
\end{itemize}

\section{A bound on self-avoiding walks}
\label{sec:saw}

In this section we derive a bound on the connective constant of the hyperbolic lattice, or equivalently 
the branching ratio of the tree of self-avoiding walks.  In particular, this bound provides the lower bound 
on $p_c$ given in Theorem~\ref{thm:pc-bounds}.

Consider a self-avoiding walk on the edges of the hyperbolic lattice $\{P,Q\}$.  
Each time we follow an edge $(u,v)$, there are $Q-1$ 
possible edges along which we could move from $v$; call these $\{1,\ldots,Q-1\}$.  
The leftmost and rightmost moves $1$ and $Q-1$ take us around the two faces to which $(u,v)$ belongs.  
If we perform $P-1$ consecutive leftmost moves, or $P-1$ consecutive rightmost ones, we would return to $u$, 
and complete the loop around one of these faces.

As a consequence, we can upper bound the entropy of self-avoiding walks by considering strings 
over the alphabet $\{1,\ldots,Q-1\}$ such that there are no runs of $P-1$ consecutive $1$s or $P-1$ 
consecutive $(Q-1)$s.  We can describe these strings with a finite-state process with $P-1$ states 
$0 \le i < P-1$, where state $i$ denotes a string ending in $i$ consecutive $1$s or $i$ consecutive $(Q-1)$s.  
The transition matrix for this process is 
\begin{align*}
M_{0,1} &= 2 \\
M_{i,0} &= Q-3 \text{ for all $0 \le i < P-1$} \\
M_{i,1} &= 1 \text{ for all $1 \le i < P-1$} \\
M_{i,i+1} &= 1 \text{ for all $1 \le i < P-1$} \\
M_{i,j} &= 0 \text{ otherwise.}
\end{align*}
We can see this as follows.  If $i=0$, so that the most recent move is some $j \notin \{1,Q-1\}$, 
there are $M_{0,1}=2$ ways to begin a run, either with a $1$ or a $Q-1$.  If the current string ends 
with a run of $i$ consecutive $1$s or $(Q-1)$s, there are $M_{i,i+1}=1$ ways to extend this run, 
$M_{i,1} = 1$ ways to switch to a run of the opposite kind, and $Q-3$ ways to end this run with a move 
$j \ne \{1,Q-1\}$.

The growth rate of these strings is the largest eigenvalue of $M$.  This is the largest root $z_0$ of its 
characteristic polynomial, which a little work shows is given by 
\begin{align*}
\det (M-z \id) 
&= (-1)^P \left( Q-3 + (Q-2) \sum_{t=1}^{P-2} z^t - z^{P-1} \right) \\
&= (-1)^P \left( \frac{Q-3 + z + (1-Q) z^{P-1} + z^P}{1-z} \right) \, ,
\end{align*}
or equivalently the largest root $z_0$ of the numerator
\[
z^P - (Q-1) z^{P-1} + z + Q-3 \, . 
\]
Since $z_0$ is an upper bound on the connective constant of the
lattice, $1/z_0$ is a lower bound on the site and bond percolation
thresholds.

\section{Technical Details}
\label{sec:technicalities}

\begin{figure}
  \centering
  \includegraphics[width=\columnwidth]{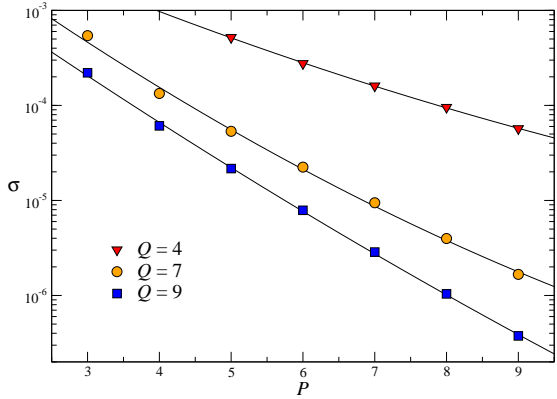}
  \vspace{-24pt}
  \caption{Decay of the standard deviation $\sigma$ of $N/B(N)$ for
    $N=12800$ vs. $P$. }
  \label{fig:stddev-P}
\end{figure}

\begin{figure}
  \centering
  \includegraphics[width=\columnwidth]{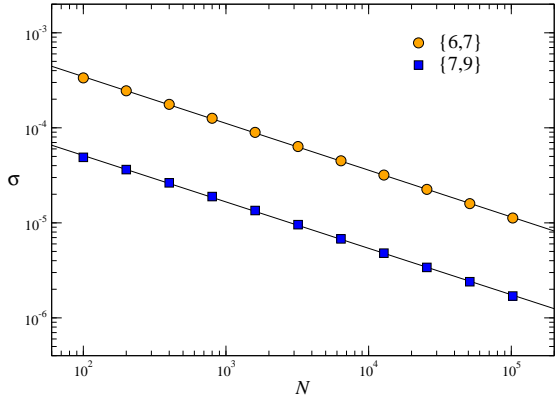}
  \vspace{-24pt}
  \caption{Decay of the standard deviation $\sigma$ of $N/B(N)$ vs. $N$. }
  \label{fig:stddev-N}
\end{figure}

For the numerical results in this paper we averaged over $10^7$
independent runs of invasion percoluation for each value of $P$, $Q$
and $N$. We computed the statistical error on $N/B(N)$ using jackknife
resampling~\cite{efron:gong:83}. The errors are of order $10^{-8}$ or
$10^{-9}$ for the larger clusters. These small values of the errors
are caused by the small values of the standard deviation $\sigma$ of
$N/B(N)$ for hyperbolic lattices. For given $N$ and $Q$, $\sigma$
decays exponentially with $P$ (Fig.~\ref{fig:stddev-P}). We know that
for a tree ($P=\infty$), $\sigma$ is zero. In addition,
$\sigma=O(1/\sqrt{N})$ (Fig.~\ref{fig:stddev-N}), which is not
surprising.  The cluster masses we used are $N=100 \cdot 2^k$ for
$k=0,1,\ldots,7$. For some systems we also simulated larger systems
with $k=8, 9, 10$ with $10^6$ samples each.  To grow an invasion
cluster of mass $N$, one needs $B(N) \simeq N/p_c$ pseudorandom
numbers. Hence each value in Table~\ref{tab:thresholds} is based upon
at least $\sim 10^{13}$ pseudorandom numbers, which were produced by
generators from the TRNG library~\cite{bauke:mertens:rng3}.  For
comparison, Americans eat roughly $10^{10}$ chickens per year
\cite{chicken:us}.

\begin{table}
  \centering
  \begin{tabular}{lccc}
  CPU & frequency & nodes$\times$cpus$\times$cores & memory/core \\[0.5ex]
  E5-1620 & 3.60 GHz & $1\times 2\times 4$ & 4.0 GByte \\
  E5-2630 & 2.30 GHz & $5\times 4\times 6$  & 5.3 GByte \\
  E5-2630v2 & 2.60 GHz &  $5\times 4\times 6$ & 5.3 GByte \\
  E5-2640v4 & 2.40 Ghz & $3\times 4 \times 10$ & 6.4 GByte
  \end{tabular}
  \caption{Computing machinery used for the simulations in this
    paper. All CPUs are Intel\textsuperscript{\textregistered} Xeon\textsuperscript{\textregistered}.}
  \label{tab:leonardo}
\end{table}

We ran our simulations on a cluster with a mixture of CPUs and a total
of 368 cores, see Table~\ref{tab:leonardo}. Each value of
Table~\ref{tab:thresholds} took roughly 24 hours wall-clock time on
this cluster. The actual invasion percolation cluster algorithm has
time complexity $O(N\log N)$, but because our labeling scheme induces
costs $O(N^{1/2})$ for handling the typical vertex in a cluster, the
total time is $O(N^{3/2} \log N)$. Memory per core can become an issue in the computation of $\pusite$ because the percolation thresholds $p_c(\hat{G})=1-\pusite(G)$ on the matching lattices are small and we need to store $B(N)\simeq N/p_c(\hat{G})$ vertices. This is the main reason why we did not go beyond $N=12800$ in the simulations for $\pusite$.

\acknowledgements{
We enjoyed fruitful discussion with Bob Ziff and we appreciate
valuable comments from Brian Hayes.  S.M. thanks the Santa
Fe Institute for their hospitality.  C.M. was supported in part by the
John Templeton Foundation and thanks Otto-von-Guericke University for their hospitality.
}

\bibliography{percolation,statmech,mertens,cs}

\end{document}